%% file: main.tex
\documentclass[journal,twocolumn, 10pt]{IEEEtran}

\usepackage{epsfig,makeidx,color,epstopdf}
\usepackage{amsmath,amssymb,amsthm,bbm}
\usepackage{cite,graphicx}
\usepackage{enumerate}
\usepackage{mathrsfs}
\usepackage{cite}
\usepackage{amsfonts}
\usepackage{mathtools, cuted}
\usepackage{lipsum}
\usepackage{indentfirst}
\usepackage{subfig}
\usepackage[font=footnotesize]{caption} 
\usepackage{csquotes}
\usepackage{float}
\usepackage{comment}

\makeatletter
\def\footnoterule{\kern-3\p@
  \hrule \@width 2in \kern 2.6\p@} 
\makeatother

\newtheorem{proposition}{Proposition}

\usepackage{titlesec}
\titlespacing*{\section}{0pt}{3mm}{0pt}
\titlespacing*{\subsection}{0pt}{2.5mm}{0pt}
\usepackage{array}
\newcolumntype{P}[1]{>{\centering\arraybackslash}p{#1}}
\newcolumntype{M}[1]{>{\centering\arraybackslash}m{#1}}

\DeclareMathAlphabet{\mathpzc}{OT1}{pzc}{m}{it}

\def\be{ \begin{equation} }
\def\ee{ \end{equation} }
\def\bea{ \begin{eqnarray} }
\def\eea{ \end{eqnarray} }

\def\b0{{\bf 0}}

\ifCLASSOPTIONonecolumn
\else
  
\fi

\makeatletter
\newcommand*\dashline{\rotatebox[origin=c]{90}{$\dabar@\dabar@\dabar@$}}
\makeatother
\makeatletter
\newcommand*{\rom}[1]{\expandafter\@slowromancap\romannumeral #1@}
\makeatother

\begin{document}
\title{Modeling RIS Empowered Outdoor-to-Indoor Communication in mmWave Cellular Networks}
\author{Mahyar Nemati, \textit{Member, IEEE}, Behrouz Maham, \textit{Senior Member, IEEE}, Shiva Raj Pokhrel, \textit{Member, IEEE} and Jinho Choi, \textit{Senior Member, IEEE}
\thanks{M. Nemati, S. R. Pokhrel, and J. Choi are with
the School of Information Technology, 
Deakin University, 
Geelong, VIC 3220, Australia (e-mail: nematim@deakin.edu.au, shiva.pokhrel@deakin.edu.au, jinho.choi@deakin.edu.au)}
\thanks{B. Maham is with the School of Engineering, Nazarbayev University, Astana, Kazakhstan (e-mail: behrouz.maham@nu.edu.kz) }
}
\date{today}
\maketitle

\begin{abstract}
With the increasing adoption of millimeter-waves (mmWave) over cellular networks, outdoor-to-indoor (O2I) communication has been one of the challenging research problems due to high penetration loss of buildings. To address this, we investigate the practicability of utilizing reconfigurable intelligent surfaces (RISs) for assisting such O2I communication. We propose a new notion of prefabricated RIS-empowered wall consisting of a large number of chipless radio frequency identification (RFID) sensors. Each sensor maintains its own bank of delay lines. These sensors which are built within the building walls can potentially be controlled by a main integrated circuit (IC) to regulate the phase of impinging signals. To evaluate our idea, we develop a thorough performance analysis of the RIS-based O2I communication in the mmWave network using stochastic-geometry tools for blockage models. Our analysis facilitates two closed-form approximations of the downlink signal-to-noise ratio (SNR) coverage probability for RIS-based O2I communication. We perform extensive simulations to evaluate the accuracy of the derived expressions, thus providing new observations and findings.
\end{abstract}
{\IEEEkeywords
Millimeter-wave (mmWave), outdoor-to-indoor (O2I) communication, reconfigurable intelligent surface (RIS).
}
\section{Introduction}
\input{Sections/1_introduction.tex}

\section{System Model Description}
\input{Sections/2_system.tex}

\section{Numerical Results}
\input{Sections/3_results}

\section{Conclusions}
\input{Sections/Conc}

\vspace{-2mm}\section*{Acknowledgement}
This work was supported by Australian Research Council (ARC) Discovery 2020 Funding, under grant number DP200100391.
%
\section*{Appendix\\Derivation of Proposition 2}
\input{Sections/5_app}


\bibliographystyle{ieeetran.bst}
\bibliography{ref}

\end{document}

%% file: Sections/1_introduction.tex

In recent years, the adoption of millimeter-wave (mmWave) spectrum has signiﬁcantly drawn attention due to its wide available bandwidth for high data rate wireless communication in fifth generation (5G) and beyond cellular systems. 
Despite the wide bandwidth availability, there are also challenges to overcome.
In particular, the mmWave communication range is often limited due to its unique propagation characteristics \cite{RAP2,JIEIoT}. 
For instance, \textit{outdoor-to-indoor} (O2I) communication has yet to be supported by mmWave networks due to the large penetration loss of building materials, such as concrete walls and metal frames. Such obstacles block mmWave signal transmissions from an outdoor base-station (BS) to an indoor user equipment (UE) inside the buildings. To this end, it is desirable to aid O2I communication to overcome the blockage and potentially improve the signal-to-noise ratio (SNR) with a newly proposed concept -- known as \emph{software controlled metasurfaces} \cite{metasurfaces}.

At a high-level, our primary focus in this work is on exploiting such a metasurface architecture that acts as an O2I refractor. It is based on our novel idea of developing wall-based fabrication consisting of several radio-frequency-identification (RFID) \cite{NematRFID} sensors. To elaborate, the chipless RFID sensors fabricated into the wall thickness with two outdoor and indoor antennas will form a system to transit the signal from O2I space. We adopt mathematics as a tool for reasoning the feasibility of the proposed system with approximations for the signal strength and O2I communication and also estimating the blockage effect. 

\subsection{What is a Software Controlled Metasurface?}
Software controlled metasurfaces have recently attracted much attention for the sixth generation (6G) of wireless communications and beyond \cite{towards,Ertrul6G}. 
These surfaces which are commonly known as \textit{reconfigurable intelligent surfaces} (RISs)\footnote{Throughout this paper, we use term RIS
to refer to any type of intelligent walls and metasurfaces.} \cite{ertrul1} have given rise to the emerging concept of \textit{smart radio environments} \cite{dismart} and are to provide an unprecedented degree of freedom in engineering wave-matter interactions for a broad range of the operational frequencies ranging from microwave to mmWave bands \cite{metasurfaces}.
RIS can turn the wireless channel, which is highly probabilistic in nature, into a controllable and
partially deterministic phenomenon \cite{RenzoRIS}. 
Specifically, RISs enable network operators to control the reflection, absorption, and refraction characteristics of the radio waves in an energy efficient-way \cite{ertrul1}. 

Each RIS is a nearly-passive smart surface including two parts: 1) passive part containing
a large number of low-cost full-duplex passive elements, i.e., \textit{unit cells}, and 2) a simple active integrated circuit (IC) acting as the wavefront controller \cite{nemati2020ris}. 
The passive RIS-elements are man-made electromagnetic sensors that are intelligently controlled by the main IC to effectively control the wavefront characteristics such as phase, frequency, amplitude, and even polarization of the impinging signals \cite{ertrul1,Ertrul6G,towards,RenzoRIS}. 
%
%
%
%
These passive sensors are always on and can reflect, refract, or even absorb the impinging signals at all time \cite{RuiRISNet}. 
In \cite{GongSurvey,GenFadRIS, liu2020, wu2020}, comprehensive overviews characterizing the performance of RIS-assisted communications affecting the propagation environments can be found. 
%
\subsection{Recent Works}
The recent advances and research of the RIS 
\cite{ergodic,R1,Mimo-1,R2,R3,AsymSINR,RuiSpa,basar2020simris} has mostly been
concentrated on its ability to reflect. 
This is reflected for example by the recent design in  analytic framework for quantifying the ergodic capacity of the RIS-aided reflection models, as investigated in \cite{ergodic}.
%
%
In \cite{R1}, the reflection impact of large-scale RIS deployments on the
cellular networks was studied. 
%
%
In \cite{R2}, an analytical probabilistic framework for successful reflection of RIS was provided. 
%
In \cite{AsymSINR}, an optimal linear precoder along with an RIS deployment in a single cell for multiple UEs was used to improve the coverage performance of the communications.
In \cite{RuiSpa}, a characterization of the spatial
throughput for a single-cell multiuser system assisted by multiple RISs that are randomly deployed in the cell was provided. It showed that the RIS-assisted model outperforms the full-duplex relay-aided counterpart system in
terms of spatial throughput when the number of RISs exceeds
a threshold.  Most recently, an investigation in \cite{basar2020simris} discussed the potential use-cases of RISs in future
wireless systems using a novel channel modeling methodology as well as a new software tool for RIS-aided mmWave networks.
%
\begin{figure}[t]
\centering
\subfloat[RIS acts as a concave mirror. ]{
 \includegraphics[width=3.5cm, height=2.6cm]{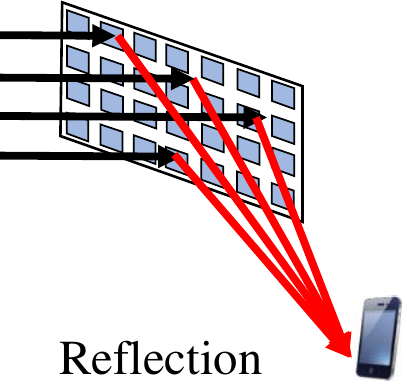}%
 }
 \hspace{4mm}
 \subfloat[RIS acts as a concave lens.]{
  \includegraphics[width=4cm, height=2.6cm]{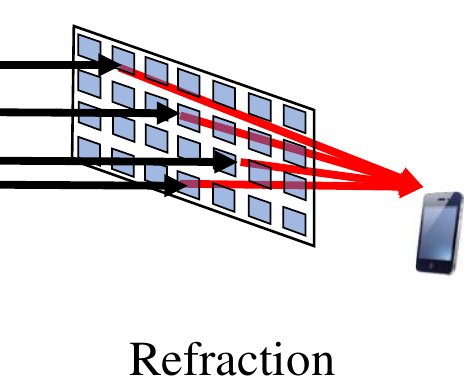}%
  }
 \captionsetup{width=1\linewidth}
 \caption{ Reflection \& refraction of the signal rays towards a UE when a) the UE is located in the same side of the serving BS and b) the UE is in the opposite side.} 
 \label{REFRAC}
\end{figure}
However, in the aforementioned studies, RIS works as a reflection surface (e.g., like a concave mirror), as shown in Fig. \ref{REFRAC} (a), while its refraction ability/use-cases studies (e.g., acting as a concave lens) are limited (as shown in Fig. \ref{REFRAC} (b)). 
Since the RISs are mostly supposed to be surfaces hanging over walls, e.g., concrete walls, a question comes to the mind that if the UE is located in the
other side of the wall, can the incident signal be refracted
through the RIS-empowered wall to reach the UE with enhanced SNR? 

From the geometrical optics perspective, both the anomalous reﬂection and refraction at the metasurface can be described by the generalized laws of reﬂection and refraction, respectively \cite{liu2020}. 
Motivated by the refraction purpose, in \cite{Refr2018_1}, it is experimentally verified that an acoustic cell prototype can provide enough degrees of
freedom to fully control the refraction angles of 60, 70, and 80 degrees.
Moreover, in \cite{perfectRef}, an overview was conducted to show that ideal refraction is feasible only if the metasurface is bianisotropic with weak spatial dispersion. This effect is described by the relations between the
exciting electric and magnetic ﬁelds and the induced polarizations in the RIS-sensors. Nevertheless, these proposed RIS prototypes are not able to overcome the building penetration loss and deliver the incident signal in mmWave band to the other side of, e.g., a concrete wall.
%
 \begin{figure}[t]
        \centering
        \captionsetup{width=1\linewidth}
        \subfloat[Large antenna array at the BS executes high-power and narrow beamforming to penetrate building materials.]{%
         \includegraphics[width=4.6cm, height=2.6cm]{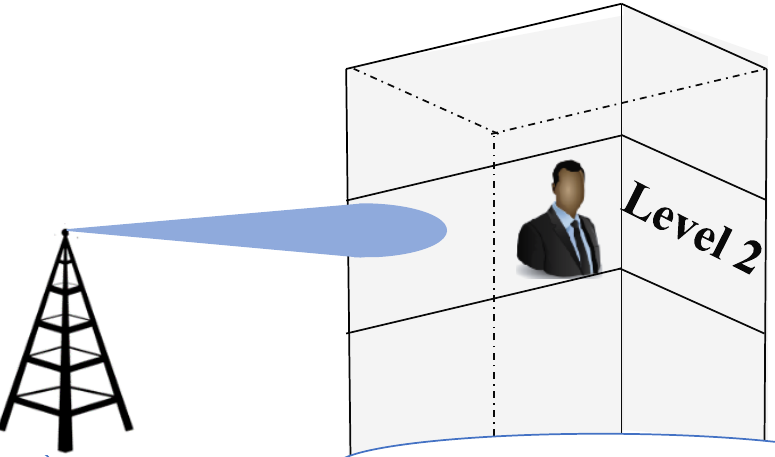}%
         } 
        \captionsetup{width=1\linewidth}
        \subfloat[Relay-aided model. BS targets the relay which transfers the signal into the indoor space using some types of low-loss conductors.]{%
         \includegraphics[width=4.6cm, height=2.8cm]{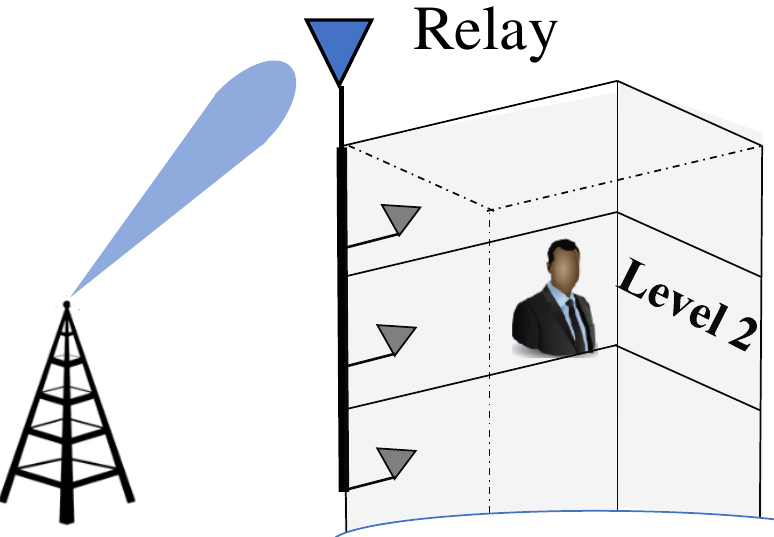}%
          \label{fig:tagB}}
        \captionsetup{width=1\linewidth}
        \subfloat[RIS-aided model. Small antenna array at the BS sends a wide beam towards the building where the UE is located at. RIS acts as a transition wall to refract the signal towards the indoor UE.]{%
         \includegraphics[width=4.6cm, height=2.6cm]{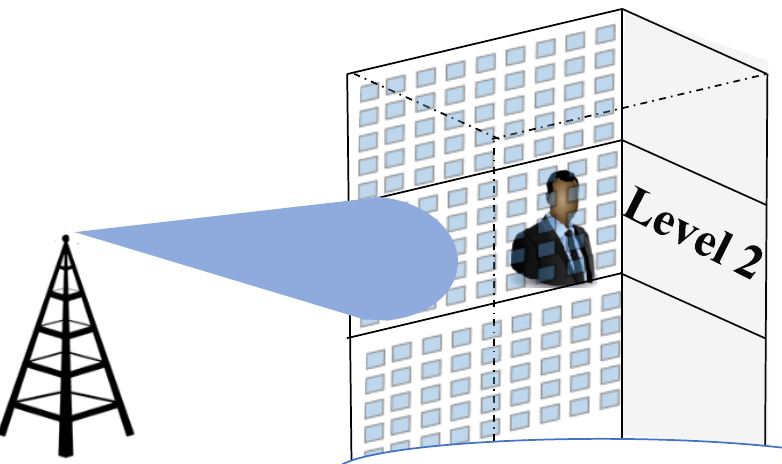}%
          \label{fig:tagB}}
         \captionsetup{width=1\linewidth}
         \caption{ Comparison between traditional, relay- and RIS-assisted O2I communication in mmWave cellular networks.}
         \label{fig:comps}
    \end{figure}
\subsection{Outdoor-to-Indoor Communication in MmWave Band}
In the mmWave band, 
 it is challenging to serve an indoor UE by an outdoor BS  because of the large penetration loss of the building materials and almost total
blockage of building walls \cite{OUT3}. 
In \cite{OUT02,OUT3,OUT4}, the indoor coverage at 
mmWave band in a
building with an outdoor BS was studied.
It is illustrated that the O2I coverage at mmWave band is quite difficult\footnote{It is noteworthy that this difficulty can be expanded to indoor-to-indoor communication. However, we only focus on O2I scenario in this study.} and the
throughput is seriously affected, depending on the wall
materials. 
Consequently, due to the large penetration loss, enabling the O2I communication needs ultra high power beams that can penetrate the building materials and reach the UE as shown in Fig. \ref{fig:comps} (a). However, forming such narrow and high power beams at the BS requires a large antenna array that might be complex and challenging and also restricted by federal laws. 

In the academia and industry, the potential proposed solution for O2I communication is relay-aided systems \cite{OUT1,Sudas}. As shown in Fig. \ref{fig:comps} (b),
in this model, which is also known as either small-cell or femto-cell concepts for coverage extension, an active relay is placed somewhere on the outdoor side of the building (usually on the roof) where it is exposed to the outdoor BS's beam. Then, the message is processed and carried by some types of low-loss conductors (e.g., fiber optic wires or other wave-guide designs \cite{OUT1}) towards the indoor transmitter antenna to be sent towards the indoor UE. The relay-aided communications are proposed in both half- and full-duplex models and require expensive hardware components. 
However, in \cite{OUT1, ertrul1,towards,Ertrul6G}, it is shown that 
full-duplex relaying has a number of drawbacks such
as signal processing complexity, noise enhancement, power
consumption and self-interference cancellations at the relay
stations. Therefore, the active relay usually operates in
half-duplex mode and is thus less spectrum efficient and still more complex and expensive than
the passive RIS that operates in full-duplex mode \cite{towards}. 
%
\subsection{Scope and Organization}
Our aim to fill the gap for O2I mmWave communication start with the feasibility analysis of a new intelligent wall architecture for intelligent transition of the signal from outdoor to indoor space as shown in Fig. \ref{fig:comps} (c); and 
then we model its performance in the presence of different blockage scenarios for both outdoor and indoor spaces. 
 %
%
 %
The contributions of this paper are summarized as follows.
\begin{itemize}
    \item We propose a new RIS architecture that acts as an O2I refractor. In particular, different from the traditional adhesive ultra-thin RIS layers, we propose a new prefabricated RIS-empowered wall which is equipped by a large number of chipless RFID sensors, operating in mmWave band.  
    Sensors are in fact built into the concrete\footnote{The proposed architecture is not limited to the concrete and can be expanded to other materials, e.g., wood, brick, etc.} wall as shown in Fig. \ref{fig:comps} (c).
    \item The chipless RFID sensors are built into the wall thickness with two outdoor and indoor antennas that can transit the signal from O2I space and converge the dispersed received signals by the wall towards an indoor UE. In fact, each RIS-sensor contains a bank of delay lines and are controlled by a main IC to adjust the phase-shift of an impinging signal and perform passive beamforming and enhance the SNR at the UE. We derive closed-form approximations for the SNR coverage probability of the RIS-assisted O2I communication model utilizing stochastic-geometry tools for blockage models.  
    \item We show that the proposed RIS-assisted model provides a diversity gain due to the wide surface of the RIS. As a result, it reduces the blockage probability significantly compared to the similar relay-aided counterpart, shown in Figs. \ref{fig:comps} (b) and (c), and increases the chance of a successful O2I communication. Moreover, the proposed model does not require complex large antenna arrays at the BS for overcoming the penetration loss of the building materials shown in Fig. \ref{fig:comps} (a). 
\end{itemize}

The rest of the paper is organized as follows. In Section II, we present the system model of the RIS-assisted O2I communication in mmWave band. Then, the principles of the blockage model are discussed in Section \rom{3}. Subsequently, the SNR coverage analysis of the RIS-assisted model is provided in Section \rom{4}. Simulation results and comparisons are discussed in Section \rom{5}. Finally, Section \rom{6} concludes the paper.

 \textit{Notation:} The 2-norm and absolute value of $\mathbf{a}$ and $a$ are denoted by $||\mathbf{a}||$ and $|a|$, respectively. $(.)^T$ and $\odot$ denote the transpose operation and element-wise multiplication, respectively. $\mathcal{CN} (\mathbf{a},\mathbf{R})$ represents the distribution of circularly symmetric complex Gaussian (CSCG) rendom vectors with mean vector $\mathbf{a}$ and covariance matrix $\mathbf{R}$. The Gaussian Q-function is given by $\mathcal{Q}(x)=\frac{1}{\sqrt{2\pi}}\int^\infty_x e^{\frac{-z^2}{2}} dz$.

%% file: Sections/2_system.tex
Suppose a single cell scenario with an outdoor BS in a distance of $R$ from a building of interest where the UE is located at as shown in Fig. \ref{fig:system}. We focus on downlink SNR coverage experienced by a UE which is equipped with an omni-directional antenna and the BS communicates with it on mmWave bands.
    Since in mmWave frequencies, signals are more disposed to be blocked by natural obstacles in the communication area, there can be different types of obstacles classified as static and dynamic blockers in the communication area as shown in Fig. \ref{fig:system}. 
    Obviously, there is no line-of-sight (LoS) link between the outdoor BS and the indoor UE. Instead there is an RIS-empowered wall crossing the wireless channel. In the following, we propose the system model in six stages.
    
    \begin{figure}[t]
            \centering
            \captionsetup{width=1\linewidth}
            \includegraphics[width=9cm, height=5.2cm]{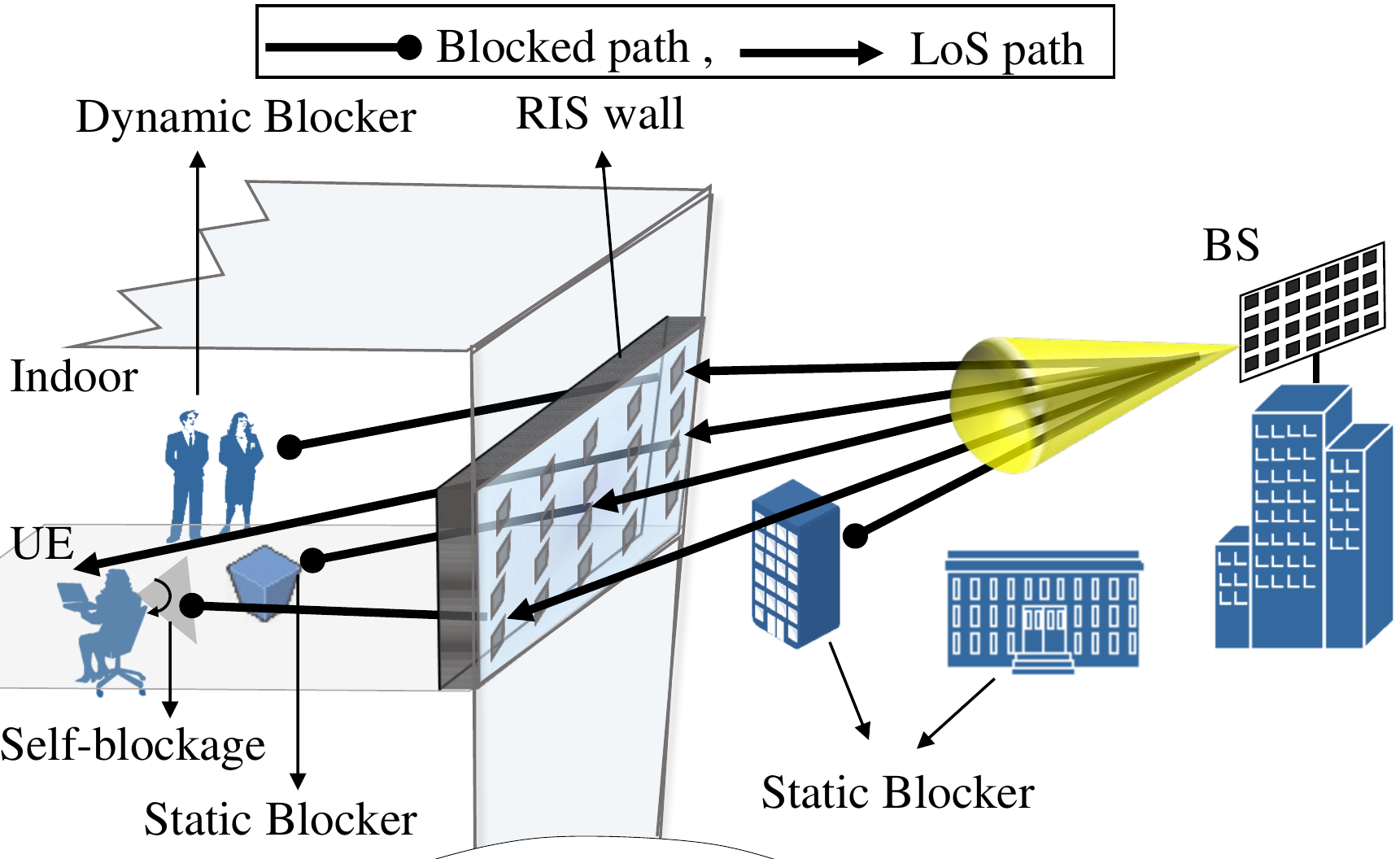}
            \caption{System model illustration.}
            \label{fig:system}
    \end{figure}

    %
 
    \subsection{Stage 1: At the BS}
    At the BS, a uniform planar square phased-array (UPA) of $M$ antenna elements, i.e., $\sqrt{M}\times \sqrt{M}$ grid, is deployed to exploit the beamforming in the mmWave band towards the RIS with a total of $N$ passive sensors.
    Let $s(t)$ and $\mathbf{w}_b\in\mathbb{C}^{M\times 1}$ represent the base-band message in the time domain and the constant weight vector at the BS, respectively. Note that $\mathbf{w}_b$ controls the gain of the beam steered towards the RIS. We assume that the peak effective radiated power (ERP) of the beam targets the center of the RIS wall. Specifically, if $r_m$ denotes the distance between the $m^\text{th}$ antenna and the center of the RIS wall, then the time-aligned transmitted signal from the $m^\text{th}$ antenna is
        \begin{align}
        s\left(t-\frac{r_m}{c}-\tau_m\right)= s\left(t-\frac{r_0}{c}+\frac{r_0-r_m}{c}-\tau_m\right),
        \label{trx}
    \end{align}
    where $c$ and $\tau_m$ represent the wave-speed and the time delay of $m^\text{th}$ antenna where $m=0,\cdots , M-1$, respectively. Additionally, $0^\text{th}$ antenna is set as a reference antenna. In fact, $\mathbf{w}_b$ entries correspond to these time delays which are associated with the phase-shifts at the BS antennas. $\tau_m$ only takes a ﬁnite number of discrete values and is given by
     $   \tau_m=\frac{r_0-r_m}{c}.$
    Thus, the transmitted beam which is a superposition of effects of $M$ antennas can be written as
    \begin{align}
        x(t)= \mathbf{w}_b s(t).
        \label{trx1}
    \end{align}
    The array gain is assumed to be constant for all angles. 
    It is noteworthy that the signal ray targeting the origin of the RIS becomes a superposition of $M$ phase-aligned signals, i.e., $x(t)\propto Ms\left(t-\frac{r_0}{c}\right)$; and consequently, the peak ERP of the beam scales up by $M^2$. 
     
    \subsection{Stage 2: BS-RIS Link}
   \begin{figure}[t]
        \centering
        \captionsetup{width=1\linewidth}
        \subfloat[Chipless RFID sensors are built into the concrete wall conducting the incident signals from outdoor to indoor space.]{%
         \includegraphics[width=8cm, height=4cm]{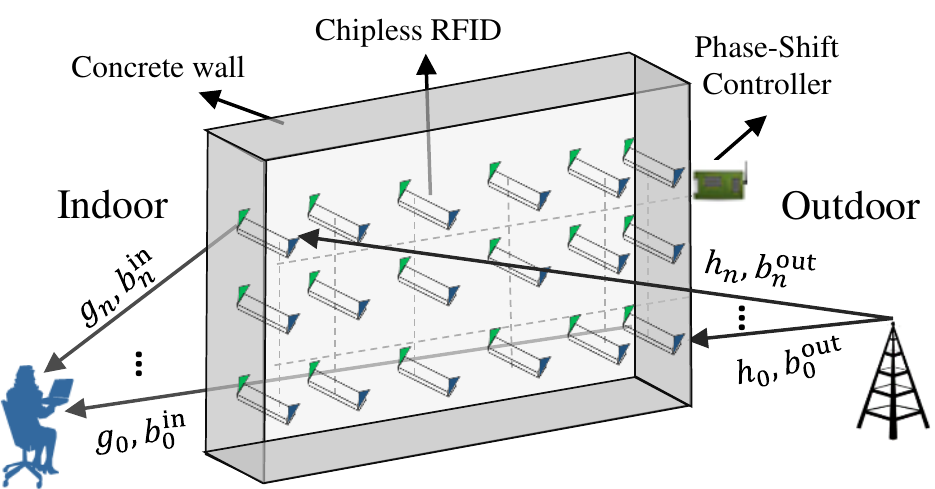}%
         } 
         \hspace{0.1\linewidth}
        \subfloat[Each chipless RFID sensor includes receive and transmit antennas along with a bank of delay lines that can be controlled by the phase-shift controller to adjust the phase-shift at the sensor and execute passive beamforming towards the indoor UE. Discrete delay lines result in discrete phase-shifts.]{%
         \includegraphics[width=5.2cm, height=3.8cm]{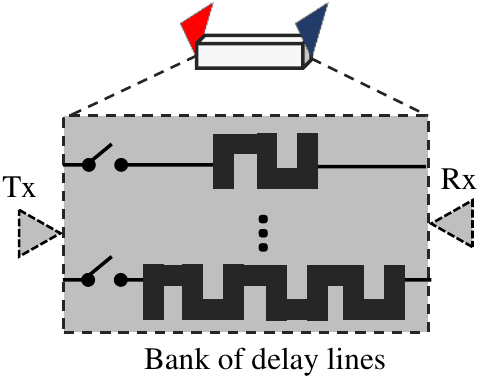}%
          }
         \captionsetup{width=1\linewidth}
         \caption{A prefabricated RIS-empowered concrete wall. }
         \label{fig:RiS}
    \end{figure}
    The signal transmitted from the BS experiences a large-scale pathloss indicated by $L=\mathcal{C}_LR^{-\alpha}$ where $\mathcal{C}_L$ and $\alpha$ are the intercept of the large-scale fading and pathloss exponent, respectively \cite{Cov&rate}.
    In addition, suppose that $\mathbf{S}_b\in \mathbb{C}^{M\times N}$ denotes the steering matrix of the BS phased-array towards the RIS-sensors.
    In fact, $\mathbf{S}_b$ is made by column vectors of $\mathbf{f}_n\in\mathbb{C}^{M\times 1}$ for $n=0,\cdots,N-1$, each refers to the steering vector towards $n^\text{th}$ RIS-sensor. 
    Then, the signal rays travel through $N$ paths with baseband equivalent channel vector of $\mathbf{h}=[h_0,\cdots, h_{N-1}]^T$ concatenated with outdoor blockage vector of $\mathbf{b}^\text{out}=[b^\text{out}_0,\cdots,b^\text{out}_{N-1}]$, $b^\text{out}_n\in \{0,\, 1\}$ to reach the RIS-sensors as shown in Fig. \ref{fig:RiS} (a), i.e., $b^\text{out}_{n}=0$ means the $n^\text{th}$ link is blocked.
    Then, the BS-RIS channel matrix, denoted by $\mathbf{H}\in \mathbb{C}^{N\times M}$, is given by
    \begin{equation}
        \mathbf{H}= [\text{diag}(\mathbf{b}^\text{out})]\times [\text{diag}(\mathbf{h})]\times \mathbf{S}_b^T.
        \label{BS1}
    \end{equation}
     We assume that the half-power beamwidth of the transmitted beam covers the RIS.
    Subsequently, the $n^\text{th}$ RIS-sensor receives a portion of the peak ERP associated with steering vector $\mathbf{f}_n$ but indicated by a factor $G_b(\phi_n,\theta_n)\in (0.5\hspace{2mm} 1]$ in direction of $(\phi_n,\theta_n)$ from the reference antenna at the BS. 
    The exact magnitude of $G_b(\phi_n,\theta_n)$ can be measured once for all sensors since the BS and the RIS are stationary. We assume that $G_b(\phi_c,\theta_c)=1$ is associated with the peak ERP targeting the central RIS-sensor. 
    As a result, the received signal at the $n^\text{th}$ RIS-sensor is given by
    \begingroup\makeatletter\def\f@size{9.5}\check@mathfonts
    \begin{equation}
        y_n(t)=\sqrt{L}\, b^\text{out}_{n} h_n \mathbf{f}_n^T x(t)=\sqrt{G_b(\phi_n,\theta_n) L}\, b^\text{out}_{n} h_n  M s(t).
        \label{ritx}
    \end{equation}
    \endgroup
    Throughout this study, the outdoor space is assumed to be a static environment with random blockages. In other words, the BS, the RIS, and reﬂecting/scattering objects are assumed to be stationary. The BS and the RIS are both at high altitude above ground so that there is no reflection from low altitude random objects on the ground as well. Thus, the path gains are assumed to either be fixed or vary slowly.
    As a result, $h_n$ for all $N$ paths can be modeled as deterministic channels. Note that the remaining random obstacles are assumed as perfect blockers without reflection/diffraction and are modeled as a blockage factor $b^\text{out}_n$.

    \subsection{Stage 3: At the RIS}
    \begin{figure}[t]
            \centering
            \includegraphics[width=6cm, height=2.8cm]{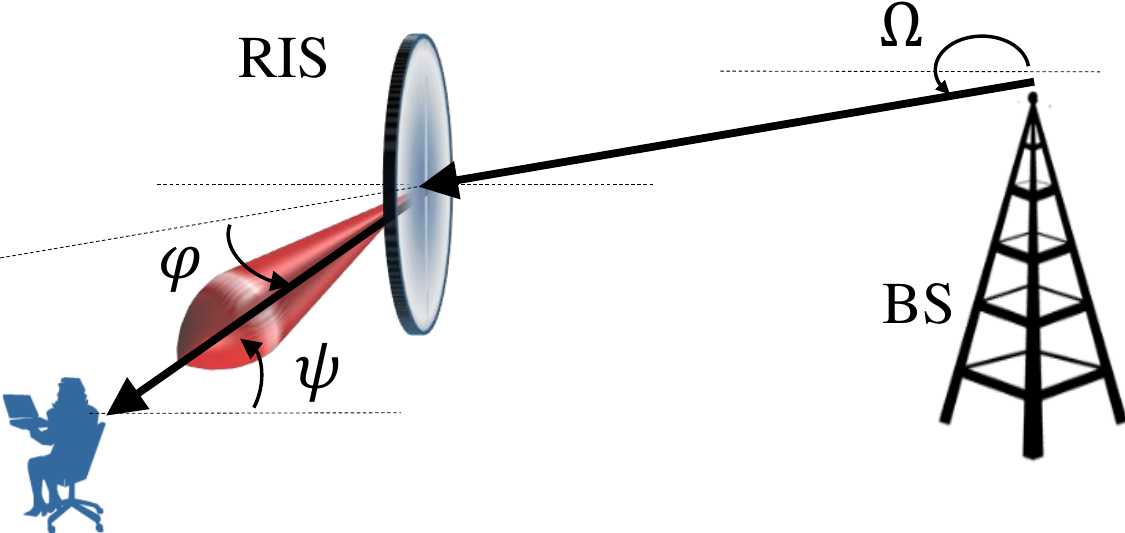}
            \caption{2-dimensional steering orientation shift at the RIS.}
            \label{beam1}
    \end{figure}
  At the third stage, the RIS receives the signal and refracts it towards the indoor UE. 
  In this study, by taking advantage of chipless RFID sensors implemented in mmWave band \cite{zomRFID,chipless19}, we propose a prefabricated RIS wall architecture where many of these passive chipless RFID sensors are built within the concrete wall as shown in Figs. \ref{fig:RiS} (a) and (b). In fact, the chipless RFID sensors that are built into the wall
thickness with two outdoor and indoor antennas are
able to transit the signal from O2I\footnote{Note that the usual RIS-elements in the literature operate in full-duplex mode since the whole
unit cell is in the same space (i.e., either in outdoor or in indoor space). However, the proposed chipless RFID sensors are built into the wall and each of its antennas operates in a different space, i.e., either outdoor or indoor. Therefore, the sensors can only operate in half-duplex mode due to separated Tx and Rx antennas. In
this study, the sensors are hired for downlink transmission. Nevertheless, similar cheap sensors can be deployed in parallel but in opposite direction for
uplink signals to provide the full-duplex transmission. Here, we just focus on downlink scenario and omit the uplink scenario without loss
of generality.} 
 space and converge
the dispersed received signals by the wall towards an
indoor UE. Motivated by this, Fig. \ref{fig:RiS} (b) illustrates that each chipless RFID sensor contains a bank of delay lines.  Next,
the phase-shift controller decides that each sensor uses a specific delay line for phase-shift adjustment to execute a passive beamforming towards the indoor UE. Here, the time delay is associated with the phase-shift.
There is extensive efforts in the
literature (e.g., in \cite{zomRFID,chipless19,CheapChip} and references 
therein) where similar passive chipless RFID sensor structures with passive components
are proposed and it is shown that the cost of these low-cost chipless RFID sensors is expected to be about 10-20 euro-cents \cite{CheapChip}. However, evaluating the electronic circuit of the
sensor in details is beyond the scope and space of this paper.

Let $\mathbf{\Lambda}\in\mathbb{C}^{N\times N}$ denote the diagonal phase-shift controller matrix at the RIS. 
Similar to the beamforming procedure at the BS, $\mathbf{\Lambda}$ diagonal entries are in fact passive beamforming coefficients that are made by the selected time delay lines in the passive RIS-sensors and are associated with phase-shifts indicated by $e^{j\Delta_n}$ for $n^{th}$ RIS-sensor. 
To be specific, $\mathbf{\Lambda}$ controls steering orientation of the passive beamforming to exploit the maximum directivity gain at the UE as shown in Fig. \ref{beam1} for a 2D space. In Fig. \ref{beam1}, $\Omega$, $\varphi$, and $\psi$ correspond to the boresight angles of the BS's beam, RIS's beam, and angle of arrival (AoA) of the signal at the UE, respectively in a 2D space. Therefore, to maximize the directivity gain at the UE, the passive beamfoming coefficients in $\mathbf{\Lambda}$ should be carefully tuned to satisfy
\begin{equation}
    \varphi = \pi + \Omega - \psi.
\end{equation}
Consequently, the dispersed received signals by the RIS wall sensors are converged together as a beam to reach the single UE's antenna. Thanks to this passive beamforming, the peak ERP of the refracted signal scales up by $N^2$. The array gain is assumed to be constant for all angles.
Note that the chipless RFID sensors are passive elements and do not amplify the power of the signal by themselves. They only adjust the phase of the signals to tune the radiation pattern in another boresight angle. Thus, the thermal noise is neglected in the impinging signal, since, ideally,
sensors do not need analog-to-digital/digital-to-analog converters, and power ampliﬁers \cite{ertrul1}.
However, there is an attenuation matrix of $\mathbf{B}=[\text{diag}([B_0,\cdots,B_{N-1}])]$ where $B_n\in [0\hspace{2mm}1)$ is the attenuation factor for the $n^{th}$ sensor. The attenuation factors can be measured once for all sensors or even can be controlled by the main RIS controller in more advanced RIS-sensor designs \cite{towards} for absorption purposes if needed. 
It is also noteworthy that each bank of delay lines in RIS-sensors only takes a finite number of discrete values due to implementation constraints \cite{nemati2020ris}. 
Eventually, the adapted signals at RIS-sensors for transition to the indoor space can be expressed as a vector given by
\begin{equation}
    \mathbf{y}=\mathbf{\Lambda} \,\mathbf{B} \sqrt{L}\, \mathbf{H} \, x(t).
    \label{y}
\end{equation}

       \vspace{-4mm}
    \subsection{Stage 4: RIS-UE Link}

    In indoor space, the large-scale pathloss is assumed to be negligible due to the regular short distances. Moreover, the RIS wavefront controller is supposed to have a passive beamforming codebook \cite{towards} to sweep its refracted beamforming coefficients from the pre-designed codebook and select the best beam based on the UE received training signal power. More advanced wavefront controller as a cognitive engine with machine learning capabilities can provide a wide range of flexibility in  intelligent beamforming. As discussed in \cite{towards}, the RIS can maintain a database that records the optimal beams for different indoor UE locations in the past and serve a new UE whose location is available. The RIS can leverage its database to efﬁciently ﬁnd an initial set of indoor beamforming coefﬁcients by using machine learning based methods \cite{towards}. 
    
    The indoor space is assumed to be a low mobility environment with random blockages. Specifically, the  RIS and reflecting/scattering objects are assumed to be stationary but blockages are random. It is justifiable with the fact that the random obstacles are assumed as perfect blockers without reflection/diffraction.
     Eventually, the refracted signal rays travel through $N$ paths with deterministic baseband equivalent channel vector of $\mathbf{g}=[g_0,\cdots,g_{N-1}]^T]$ but with indoor blockage vector of $\mathbf{b}^\text{in}=[b^\text{in}_0,\cdots,b^\text{in}_{N-1}]^T$, $b^\text{in}_n\in \{0,1\}$ and reach the single UE antenna as shown in Fig. \ref{fig:RiS} (a).
     Therefore, the RIS-UE link can be represented by vector $\mathbf{q}$ given by
     \begin{equation}
         \mathbf{q}= \mathbf{b}^\text{in} \odot \mathbf{g},
         \label{G}
     \end{equation}
    and from \eqref{trx1}\eqref{y}, the received signal at the UE with the background noise $w(t)\sim\mathcal{CN}(0,\sigma_{w}^2)$ can be written as
    \begin{align}
        z(t)&=\mathbf{q}^T\, \mathbf{y}+w(t)\nonumber\\
        &= \mathbf{q}^T\, \mathbf{\Lambda} \,\mathbf{B}  \sqrt{L}\, \mathbf{H} \, \mathbf{w}_b s(t)+ w(t).
        \label{eq21}
    \end{align}
    It is noteworthy that the deterministic independent \textit{Nakagami-m} small-scale fading is assumed for both outdoor and indoor links to easily adapt the different degree of fading by changing the value of \textit{m}-parameter. We assume \textit{m}-parameter is a positive large integer number to approximate the small-variance fading due to the nature of LoS links \cite{Cov&rate}. Thus, $|h|^2$ and $|g|^2$ are normalized Gamma deterministic variables.  Additionally, profiting from smart beamforming at the RIS, the delay of paths at the UE is negligible and the impact of frequency-selective fading can be neglected using advanced OFDM numerology design \cite{Licis} or frequency domain equalization techniques \cite{Cov&rate}. 

          Consequently, considering normalized message power, i.e., $\mathbb{E}\left\{|s(t)|^2\right\}=1$, and the RIS phase compensation, i.e., $|h_n e^{j\Delta_n} g_n|^2=|h_n|^2|g_n|^2$, the SNR at the UE becomes
    \begin{equation}
        \Gamma =\frac{M^2L\sum\limits_{n=0}^{N-1}B_nG_b(\phi_n,\theta_n)|h_n|^2|g_n|^2|(b^\text{out}_{n}\, b^\text{in}_{n})^2}{\sigma_w^2}.
        \label{rt2}
    \end{equation}

   \subsection{Stage 5: Channel State Information} It is worthwhile to note that since random obstacles are assumed to block a link without diffraction or reflection in a static environment, the composite channel coefficients of $h_n g_n$ for all $N$ paths can be estimated\footnote{Using the ﬁxed pilot pattern in LTE standard for slow fading channel
estimation in \cite{nutshel}.} once to create a pre-designed passive beamforming codebook at the RIS controller accordingly. Then, the RIS wall with its sensors profits from a simple passive AoA localization \cite{AOA,PuOB} to find out the direction of the UE's training signal and select the best beam for the intended indoor UE and subsequently perform a phase-shift adjustment at its sensors. In other words, given beam sweeping measurements, the beamforming vector can be chosen from the pre-designed beam sweeping codebook.   

   \subsection{Stage 6: Double Pathloss Effect of RIS Deployment} In a generic RIS deployment scenario, due to the passive nature of the RIS-sensors, the signal suffers from double large-scale pathloss; while profiting from the passive beamforming at the RIS which causes the power of the refracted beam scales up by $N^2$. Thus, the received power at the UE, denoted by $P_u$, follows 
\begin{equation}
    P_u \propto \frac{N^2 }{ d_\text{out}^{\alpha}d_\text{in}^{\alpha}},
\end{equation}
where $d_i$, $i\in\{\text{out,in}\}$ stands for BS-RIS and RIS-UE distances.
As observed in \cite[Fig. 21]{wu2020}, placing the RIS near the UE yields the highest received power and minimizes the double pathloss effect; while placing it around the middle between the UE
and BS (usually optimal in the case of an active relay instead of RIS) leads to the smallest received power. Intuitively speaking, in our proposed RIS-assisted O2I communication, the RIS-wall is usually very close to the UE compared the BS, i.e., $d_\text{in}\ll d_\text{out}$, which minimizes the double pathloss effect. Additionally, due to the short distances in the indoor space, the indoor large-scale pathloss is supposed to be negligible.
%
It is worth mentioning that minimization of the double pathloss effect is not an additional assumption, but the beauty of our key proposed RIS-assisted O2I communication model. 
   
In the following section, we explain the principles of the blockage model in more details.

    \section{Blockage Model for O2I Communication Deployment}
    Blockage results from three natural phenomena such as static, dynamic, and self blockage \cite{MOBLCEL,MOBLCEL-1,MOBLCEL-2} as shown in Figs. \ref{fig:system} and \ref{fig:beam}:
    \begin{itemize}
    \item Static blockage is a result of stationary objects like buildings in the communication area;
    \item Dynamic blockage is due to mobile objects like humans;
    \item Self blockage is caused by user's body orientation which can block a link. 
    \end{itemize}
       \begin{figure}[t]
            \centering
            \includegraphics[width=7cm, height=4cm]{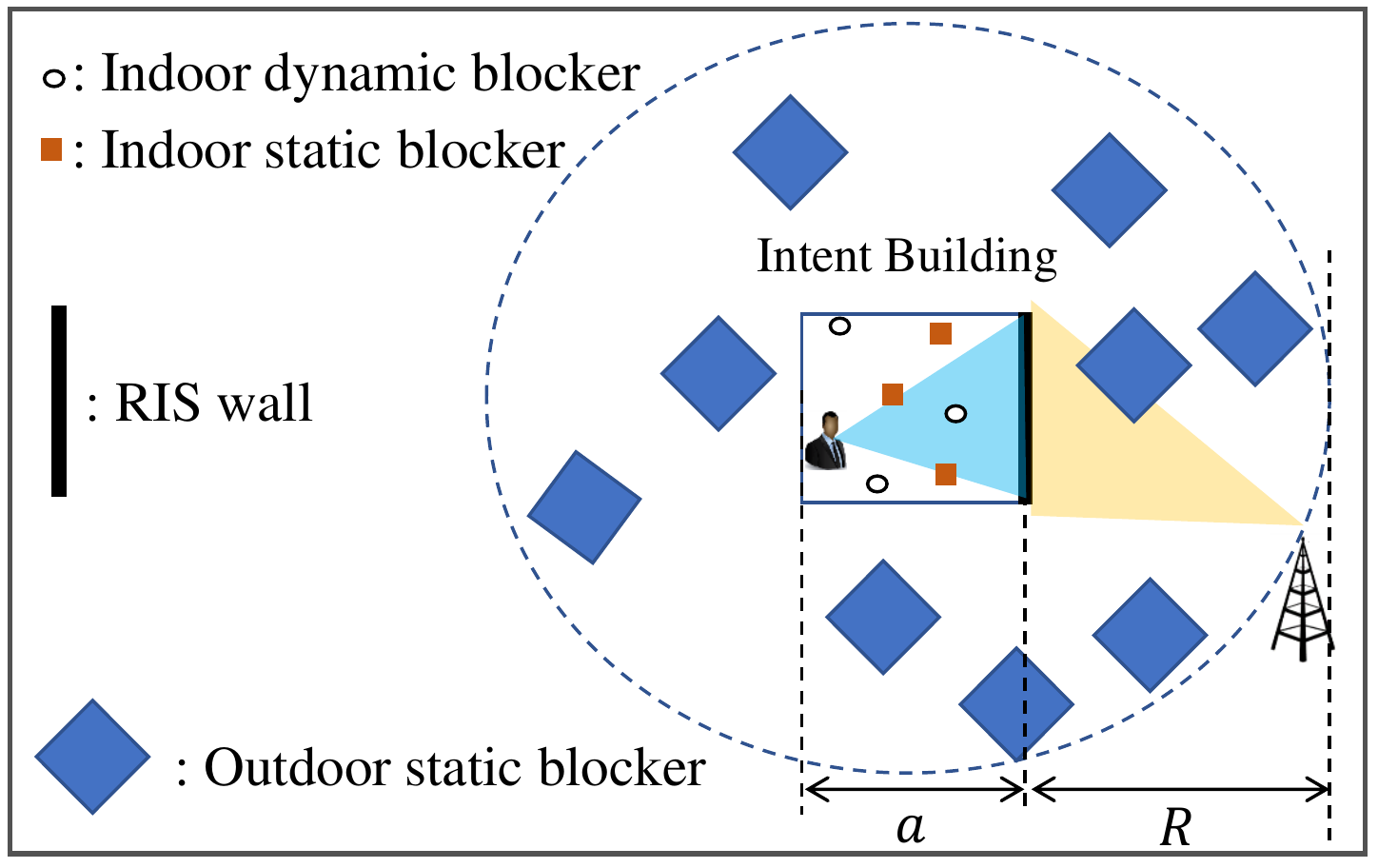}
                \caption{Depiction of UE, BS, and blockers.}
            \label{fig:beam}
    \end{figure}
    As shown in Fig. \ref{fig:system}, we assume that the outdoor BS is at high altitude and free of the outdoor dynamic blockage effect caused by low altitude mobile objects like humans and cars. Therefore, it is a feasible to assume only static blockage in the outdoor BS-RIS link.
     On the other hand, in the indoor area, both static and dynamic blockages are likely to happen. Here, 
     let $\Phi_\text{st}^\text{out}$, $\Phi_\text{st}^\text{in}$, and $\Phi_\text{dy}^\text{in}$ denote the distributions of outdoor static, indoor static, and indoor dynamic blockers, respectively. These three distributions of blockers are assumed to be independent homogeneous Poisson Point Processes (PPP).
    Furthermore, $\lambda^\text{out}_\text{st}\left[\frac{\text{bl}}{\text{km}^2}\right]$, $\lambda^\text{in}_\text{st}\left[\frac{\text{bl}}{\text{m}^2}\right]$, and $\lambda^\text{in}_\text{dy}\left[\frac{\text{bl}}{\text{m}^2}\right]$ are blockers intensities in $\Phi_\text{st}^\text{out}$, $\Phi_\text{st}^\text{in}$, and $\Phi_\text{dy}^\text{in}$, respectively.
    Thus, the probabilities of having $m^i_{j}$ blockers in an area of $A_i$ is
        \begin{align}
          \Pr[m^i_{j}]=\frac{\left[\lambda^i_{j} A_i\right]^{m^i_{j}}}{m^i_{j}!}e^{-\lambda^i_{j} A_i},
    \label{er1}
    \end{align}
    %
    where index of $i\in\{\text{out},\text{in}\}$ stands for outdoor/indoor spaces. Additionally, index of $j\in\{\text{st},\text{dy}\}$ indicates the type of blockers, i.e., $\text{st}$: static, $\text{dy}$: dynamic.
    For instance, in Fig. \ref{fig:beam}, the outdoor and indoor areas are $A_\text{out}=\pi \left(R+\frac{a}{2}\right)^2-a^2$, and $A_\text{in}=a^2$, respectively; where $a$ is the square side of the intent building.
   Moreover, the outdoor and indoor blockage factor for $n^\text{th}$ link can be modeled as a Bernoulli random variable (r.v.) which is given by
 \begin{equation}
        b^i_{n}=\left\{\begin{array}{lll}
            0, & {\rm w.p.} \ p_{i,n} & \text{{\small (blockage)}} \\
            1, & {\rm w.p.}\ 1-p_{i,n} & \text{{\small (no blockage)}}
        \end{array} \right. ,
    \end{equation}     
    where $p_{i,n}$, $i\in\{\text{out, in}\}$ denotes the blockage probability for $n^{th}$ link.
    \subsection{Outdoor Blockage Model}
    The outdoor blockage is a result of outdoor static blockers. For instance, buildings and trees are such static blockers in the outdoor area. Note that between the BS and the RIS, there is no self-blockage and low altitude mobile objects like cars and humans are neglected due to relatively high altitude communication. Therefore, $p_{\text{out},n}$ is given by
    \begin{equation}
        p_{\text{out},n}=1-\Pr[b_{n}^\text{out}\stackrel{\text{static}}{=}1].
        \label{ob1}
    \end{equation}
    We assume that the static blockers are located randomly represented by the process of random rectangles in \cite{BlokPro,Cov&rate}. 
    Subsequently, LoS probability of $\Pr[b^\text{out}_{n}\stackrel{\text{static}}{=}1]$, using the void probability in Poisson process \cite{BlokPro}, is given as follows 
    \begin{equation}
     \mathcal{P}_{\text{st},n}^\text{out}=   \Pr[b^\text{out}_{n}\stackrel{\text{static}}{=}1]= e^ {-\eta_1(\kappa_1\mathpzc{R}_{1,n}+\upsilon_1)},
        \label{ew3}
    \end{equation}
    where $\kappa_{1}=\frac{2\lambda^\text{out}_\text{st}}{\pi}\left(\mathbb{E}\{\mathcal{L}\}+\mathbb{E}\{\mathcal{W}\}\right)$ and $\upsilon_1=\lambda^\text{out}_{st}\mathbb{E}\{\mathcal{L}\}$ $\mathbb{E}\{\mathcal{W}\}$ in which $\mathbb{E}\{\mathcal{L}\}$ and $\mathbb{E}\{\mathcal{W}\}$ are the average of length and width of the static blockers. $\mathpzc{R}_{1,n}$ is the 2D distance of BS from $n$th RIS-sensor.
    $\eta_1$ is a constant scaling factor incorporating the height of blockers given in \cite[Eq. (7)]{BlokPro}. In other words, it denotes conditional probability that the static blocker crossing a link has enough height to block the link.

    \subsection{Indoor Blockage}
    \begin{figure}[t]
            \centering
            \includegraphics[width=6cm, height=2cm]{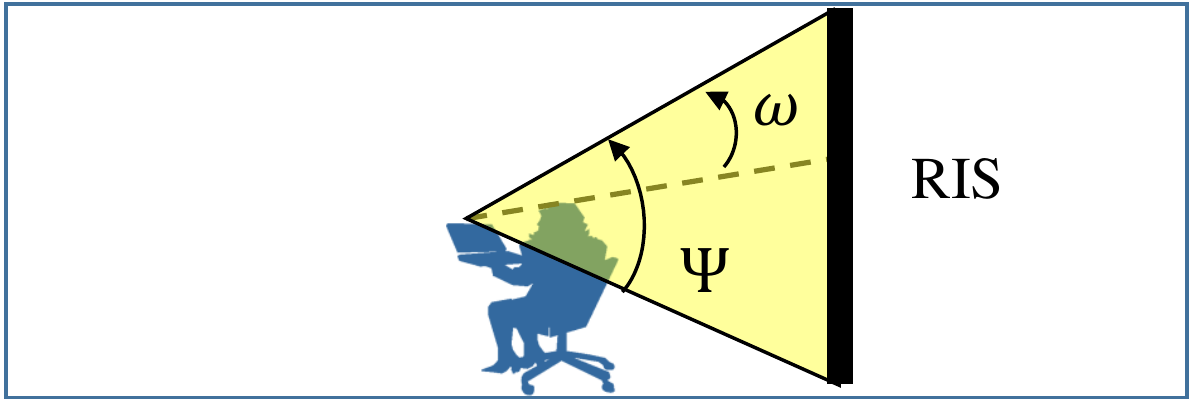}
                \caption{Illustration of self blockage.}
            \label{fig:self}
    \end{figure}
    The indoor blockage is a result of static, dynamic and self blockages. Static blockers are such indoor fixed obstacles, e.g., desks, and dynamic blockers are like other people moving around. Following \cite{MOBLCEL} and \eqref{ob1}, $p_{\text{in},n}$ can be given by
    \begin{equation}
         p_{\text{in},n}=1-\Pr[b^\text{in}_{n}\stackrel{\text{self}}{=}1]\Pr[b^\text{in}_{n}\stackrel{\text{static}}{=}1]\Pr[b^\text{in}_{n}\stackrel{\text{dynamic}}{=}1]. 
    \end{equation}
    As shown in Fig. \ref{fig:self}, $\Pr[b_{n}^\text{in}\stackrel{\text{self}}{=}1]$ becomes
    \begin{equation}
      \mathcal{P}_{\text{self},n}^\text{in}= \Pr[b_{n}^\text{in}\stackrel{\text{self}}{=}1]=\left(\frac{\omega}{\Psi}\right),
    \end{equation}
    where $\omega \in [0,\quad \Psi]$ is the angle of user's body in which signal is not blocked.
    In addition, similar to the outdoor static blockage explained in previous subsection, the LoS probability of $\Pr[b^\text{in}_{n}\stackrel{\text{static}}{=}1]$ becomes 
    \begin{equation}
        \mathcal{P}_{\text{st},n}^\text{in}=\Pr[b^\text{in}_{n}\stackrel{\text{static}}{=}1]= e^ {-\eta_2(\kappa_2\mathpzc{R}_{2,n}+\upsilon_2)},
        \label{ew5}
    \end{equation}    
    where definitions of $\eta_2$, $\kappa_2$ and $\upsilon_2$ are the same as those of $\eta_1$, $\kappa_1$ and $\upsilon_1$, respectively, but with respect to the indoor static blocker parameters. In \eqref{ew5}, $\mathpzc{R}_{2,n}$ denotes the 2D distance between $n^\text{th}$ RIS-sensor and the UE's antenna.
        Moreover, there exist dynamic blockers that move in random directions with  the average speed of $V_\text{in}$ and may cross the communication links and block them. As observed in \cite{MOBLCEL-1}, the arrival of dynamic blockers at the $n^\text{th}$ link has Poisson distribution with intensity of $\beta^\text{in}_{n}$ (bl/sec) and the blockage duration has exponential distribution with mean $\frac{1}{\mu^\text{in}}$ (sec). Therefore, the average number of dynamic blockers that block $n^\text{th}$ link at same time is $\frac{\beta^\text{in}_{n}}{\mu^\text{in}}$ (bl).
Particularly, $\beta^\text{in}_{n}$ (bl/sec) and $\mu^\text{out}$ (1/sec) are associated with average blocked and unblocked rates, respectively \cite{MOBLCEL}. 
    Consequently, following \cite[Eq. (5)]{MOBLCEL-2}, LoS probability of $\Pr[b^\text{in}_{n}\stackrel{\text{dynamic}}{=}1]$ becomes
    \begin{equation}
     \mathcal{P}_{\text{dy,n}}^\text{in}=  
     \Pr[b^\text{in}_{n}\stackrel{\text{dynamic}}{=}1]=\frac{\mu^\text{in}}{\beta^\text{in}_{n}+\mu_\text{in}},
     %
        \label{ew2}
    \end{equation}
    where $\beta^\text{in}_{n}$ is given by 
    \begin{equation}
        \beta^\text{in}_{n}=
             \frac{2}{\pi} \lambda^\text{in}_\text{dy} V_\text{in} \frac{H^\text{in}_\text{bl}-H_\text{ue}}{{H}_n-H_\text{ue}} \mathpzc{R}_{2,n},
             \label{bet2}
    \end{equation}
     In above expression, $H^\text{in}_\text{bl}$, ${H}_{n}$, and $H_\text{ue}$ are the heights of the indoor dynamic blockers, $n^\text{th}$ RIS-sensor, and the UE, respectively. 

   \begin{figure}[t]
            \centering
            \includegraphics[width=1\linewidth, height=6cm]{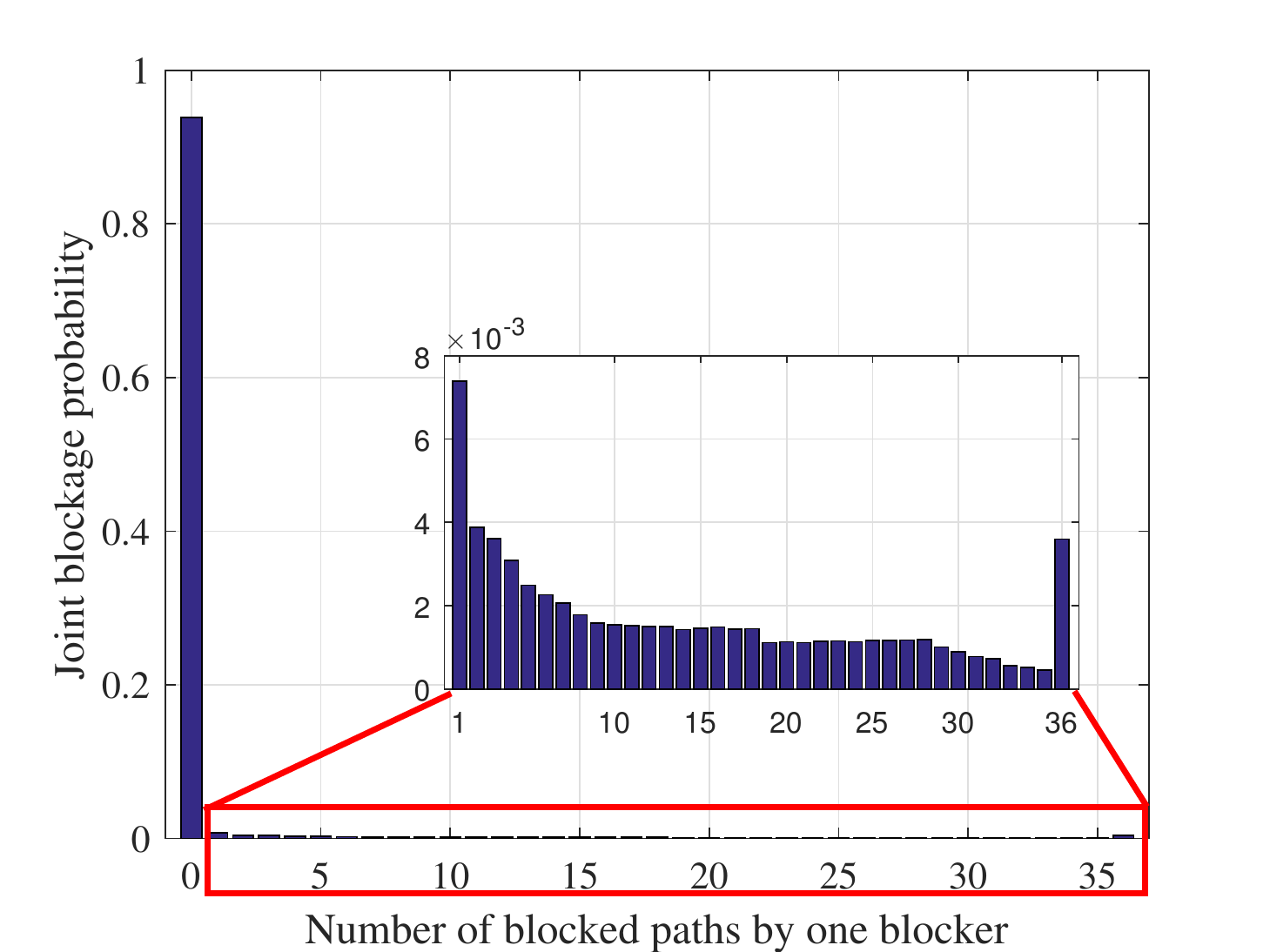}
                \caption{Joint blockage probability of outdoor static blockage effect where $N=36$, $\eta_1=0.5$, $\mathbb{E}\{\mathcal{W}\}=\mathbb{E}\{\mathcal{L}\}=10$m, $\lambda_\text{st}^\text{out}=25\frac{\text{bl}}{\text{km}^2}$, and $R=60$m.}
            \label{fig:corel_0}
    \end{figure}    
      
    \subsection{Towards Sophisticated Blockage Model}

           \begin{figure}[t]
            \centering
            \includegraphics[width=7cm, height=2.7cm]{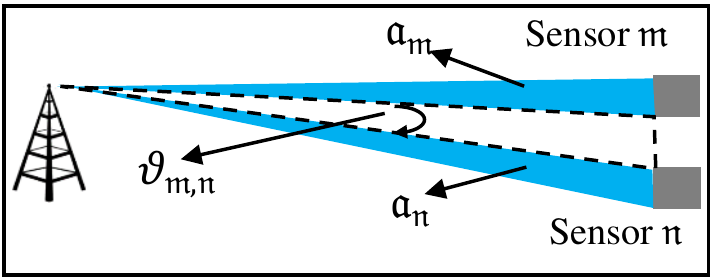}
                \caption{Blockage regions for paths $\mathfrak{n}$ and $\mathfrak{m}$.}
            \label{fig:corel}
    \end{figure}

    Our earlier analyses of LoS or blockage probabilities presumed independence between different links which has been one of the basic assumption in most of the existing works~\cite{BlokPro, Cov&rate, MOBLCEL-1,MOBLCEL,MOBLCEL-2}. However, the probabilities for different links are not independent in real networks and there are potential correlations of blockage effects between links specially when there exists a large static blocker. This correlation is effective here where there are a large number of co-located RIS-sensors along with a single BS. For instance, a group of the paths might be blocked by a large obstacle simultaneously. 
    Fig. \ref{fig:corel_0} shows a conducted numerical result using MATLAB for a specific scenario that indicates the joint blockage probabilities 
    of the paths are in range of $\sim 10^{-3}$ for the square array of RIS-sensors with minimal distance of $1$m between them. 
    As illustrated for $\mathfrak{m}^\text{th}$ and $\mathfrak{n}^\text{th}$ sensors in Fig.~\ref{fig:corel}, $\mathfrak{m}\neq\mathfrak{n}$, when the separation angle between the paths, i.e., $\vartheta_{\mathfrak{m},\mathfrak{n}}$, increases, the probability that the paths being blocked with the same blocker decreases.
    Note that a perfect blockage in $\mathfrak{n}^\text{th}$ path happens when the blocker obstructs the area of $\mathfrak{a_n}$.
     Intuitively speaking, the separation angle between paths justifies the joint probability reduction in Fig. \ref{fig:corel_0} when the number of jointly blocked paths increases as well as the separation angle between them increases.
       In Fig. \ref{fig:corel_0}, the jump in the joint blockage probability when the whole $N=36$ paths are blocked is because of the fact that separation angle is zero for all paths. 
    It explains that a more sophisticated blockage model with multiple channel states would require for accurate tracking of the evolution of the channel state for each path and their correlation analysis\footnote{For tracking of the evolution of the channel state for each path and their correlation analysis, a multidimensional Markovian approach would need to be developed in future work.}. %
    Such a correlation between the blockage of different paths can be approximated and considered partially as in~\cite{gupta2017macrodiversity, corel}. Along the similar lines of~\cite{UrbanCor,gupta2017macrodiversity, corel}, we can consider a  set  of $N$ paths  and  take  into  account  the  correlation  of the blockage across paths.
    Here, each path can have $2^N-N-1$ correlation coefficients with other paths, i.e., proper subsets of $\{1,\cdots,N\}$ excluding $N$ single-element subsets. 
    For instance, let us assume the rectangular height for the static blockage, i.e., a 3D object, along with its orientation angle, denoted by $\Theta$, in a second-order-diversity outdoor static blockage model. Then, following the analysis in \cite{gupta2017macrodiversity} and using void probability in Poisson process, we have
      \begin{align}
      \mathbb{E}\{b^\text{out}_\mathfrak{m} b^\text{out}_\mathfrak{n}\}
      &=
        \Pr\big[b^\text{out}_\mathfrak{m}\stackrel{\text{static}}{=}1, b^\text{out}_\mathfrak{n}\stackrel{\text{static}}{=}1\big] \nonumber\\
        &=
        \exp\left(- \lambda_\text{st}^\text{out}  \int_{\Theta}\int_{ (\mathbb{R}^+)^3}
        \left(\mathfrak{a}_\mathfrak{m}  \cup \mathfrak{a}_\mathfrak{n}\right)\right),
        \label{cor_0}
    \end{align}
    which is equivalent to the joint LoS probability when all blockages are outside the shaded area in Fig. \ref{fig:corel}.  
    %

    Let $\rho_\text{st}^\text{out}(\mathfrak{m},\mathfrak{n})$ denote the second-order-diversity outdoor static LoS correlation coefficient between paths $\mathfrak{m}$ and 
    $\mathfrak{n}$, ($\mathfrak{m},\mathfrak{n}=0,\cdots,N-1,\,\mathfrak{m}\neq 
    \mathfrak{n}$), which can be estimated as 
    \begin{align}
        \rho_\text{st}^\text{out}(\mathfrak{m},\mathfrak{n}) &=\frac{\mathbb{E}\{b^\text{out}_\mathfrak{m}b^\text{out}_\mathfrak{n}\}-\mathbb{E}\{b^\text{out}_\mathfrak{m}\}\mathbb{E}\{b^\text{out}_\mathfrak{n}\}}{\sqrt{\sigma_{\mathfrak{m}}^2\sigma_{\mathfrak{n}}^2}}
        \nonumber\\
        &=
        \frac{\Pr[b^\text{out}_\mathfrak{m}\stackrel{\text{static}}{=}1,b^\text{out}_\mathfrak{n}\stackrel{\text{static}}{=}1]-
        \mathcal{P}^\text{out}_{\text{st},\mathfrak{m}}
        \mathcal{P}^\text{out}_{\text{st},\mathfrak{n}}}
        {\sqrt{\mathcal{P}^\text{out}_{\text{st},\mathfrak{m}}
        \mathcal{P}^\text{out}_{\text{st},\mathfrak{n}}
        (1-\mathcal{P}^\text{out}_{\text{st},\mathfrak{m}})
        (1-\mathcal{P}^\text{out}_{\text{st},\mathfrak{n}})}}
        \label{cor_1}.
    \end{align}
     However, \eqref{cor_0} and \eqref{cor_1} 
    are for the second-order statistics; 
    while for an accurate analysis, higher order statistics 
    need to be derived which are difficult to obtain. 
    Hence, Fig. \ref{ProdDis0} depicts the numerical results of end-to-end blockage probabilities 
    with/without the blockage correlation on blockage. 
    It is shown that the gap between the correlated and independent 
    blockage models increases when 
    either end-to-end distance or blockage density increases. 
    Accordingly, the following remark concludes this subsection.

      \begin{figure}[t]
            \centering
            \includegraphics[width=1\linewidth, height=6cm]{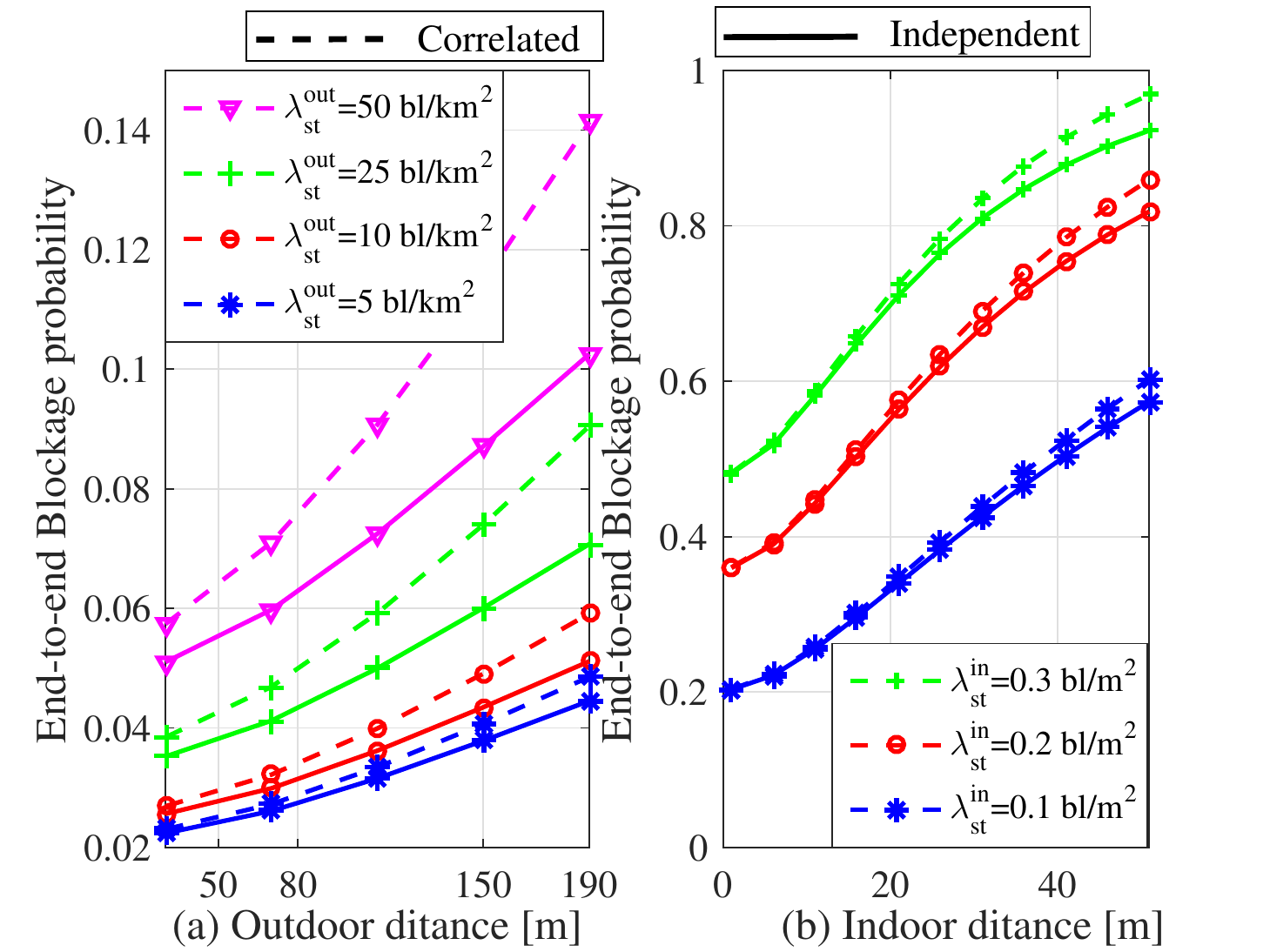}
                \caption{Evaluation of the blockage with respect to the distance. 
         \textemdash ~Outdoor parameters: 
         $\mathbb{E}\{\mathcal{W}\}=\mathbb{E}\{\mathcal{L}\}=10$m, $\eta_1=0.5$,  
          ~Indoor parameters: $H^\text{in}_\text{bl}=2$m, $\mathbb{E}\{\mathcal{W}\}=\mathbb{E}\{\mathcal{L}\}=0.5$m, $\eta_2=0.25$, $\lambda_\text{dy}^\text{in}=0.1\frac{\text{bl}}{\text{m}^2}$, $\mu_\text{in}=1$s, $V_\text{in}=0.5\frac{\text{m}}{\text{s}}$, $N=36$.} 
         \label{ProdDis0}
    \end{figure}    
     \textit{Remark 1:} As numerically shown in Fig. \ref{ProdDis0}, when the distance between two ends increases, the impact of blockage correlation on the blockage probability becomes larger.
    This is because of the fact that blockage probability, joint blockage probability, and blockage correlation increase when the distance increases. Intuitively speaking, when the end-to-end distance increases, e.g., $R\rightarrow \infty$, a narrower beamforming is required and therefore, all highly correlated $N$ paths can be treated as one single path.
     However, in smaller distances shown in Fig. \ref{ProdDis0}, i.e., $\sim80$m for outdoor and $\sim20$m for indoor, the blockage correlation model approaches the independent blockage model (preferably for small $\lambda_\text{st}^\text{out}$).
      Moreover, with these reasonably short distances, the correlation between paths becomes negligible when size of blockers are relatively small compared to the end-to-end points \cite{BlokPro,Cov&rate} and ignoring the correlation of shadowing between links causes minor loss in accuracy.
    
    Without loss of generality, throughout this paper, for mathematical tractability and similar to the works in~\cite{BlokPro, Cov&rate, MOBLCEL-1,MOBLCEL,MOBLCEL-2}, we assume that the correlation is negligible in the aforementioned short distances (as it does not affect the conclusions of this study). Moreover, further investigation of blockage correlation effect is deferred to the future work.
    Besides, we assume that the size of blockers is small enough that the blockage correlation effect of them is negligible.

    \section{SNR Coverage Analysis}
    The SNR coverage probability is the probability that the received SNR is larger than a threshold. Let $T$ denote the threshold, then from \eqref{rt2}, the SNR coverage probability becomes
    \begin{equation}
    \Pr\left[\Gamma>T\right]=1-\Pr\{\Gamma\leqslant T\}.
    \end{equation}
     It is equivalent to the complementary cumulative distribution function (CCDF) of SNR. In \eqref{rt2}, for notational simplicity, let $Z_n=b^\text{out}_{n}b^\text{in}_{n}$, $\mathcal{A}_n=B_nG_b(\phi_n,\theta_n)|h_n|^2|g_n|^2$, 
     and $\mathcal{G}=\frac{M^2 L}{\sigma_{w}^2}$. Since $Z_n\in\{0,\,1\}$, $Z_n^2=Z_n$ and \eqref{rt2} can be re-expressed as 
    \begin{equation}
        \Gamma =\mathcal{G}\sum_{n=0}^{N-1}\mathcal{A}_nZ_n.
        \label{rt22}
    \end{equation}
    Then, the probability of blockage for the $n^{th}$ end-to-end path between the BS and UE becomes
    \begin{align}
        \mathfrak{p}_n&=\Pr\{Z_n=0\}=1-\Pr\{Z_n=1\}\nonumber\\
        &=1-\left[1-p_{\text{out},n}\right] \left[1-p_{\text{in},n}\right]\nonumber \\
        &=p_{\text{out},n}+p_{\text{in},n}-p_{\text{out},n}p_{\text{in},n}.
        \label{eqpn}
    \end{align}
    With respect to \textit{Remark 1}, the blockage distribution of each equivalent link ($Z_n$) is defined as independent Bernoulli r.v. with parameter $\mathfrak{p}_n$ given in (\ref{eqpn}). Therefore, $\Gamma$ has a weighted sum of independent Bernoulli trail. 
    Since the probability mass function (PMF) of the sum of weighted Bernoulli r.v.'s is complicated \cite{WeiBook,sumind}, in order to find a closed-form expression for the SNR coverage probability, we consider two different approximations as follows.
    
    \subsection{\textbf{Approximation-\rom{1}}}
    In this approximation, we consider a wide RIS-wall containing a large number of sensors. Therefore, we state the following proposition.
    
    \begin{proposition}
    The O2I SNR coverage probability when the RIS-wall contains a large number of sensors is given by 
    \begin{equation}
         \Pr\left[\Gamma>T\right]=\mathcal{Q}\left(\frac{T-\mathcal{M}}{\sigma_z}\right),
    \end{equation}
    where $\mathcal{M}$ and $\sigma^2_z$ are the mean and variance of $\Gamma$ given as
      \begin{align}
            &\mathbb{E}\left\{\Gamma\right\}=\mathcal{G}\sum_{n=0}^{N-1}\mathcal{A}_n(1-\mathfrak{p}_n)=\mathcal{M},\\
            &
            Var\{\Gamma\}=\mathcal{G}^2\sum_{n=0}^{N-1}\mathcal{A}_n^2(1-\mathfrak{p}_n)\mathfrak{p}_n=\sigma_z^2.
        \end{align}
    \end{proposition}
    
\begin{proof}
In the literature, the RIS is mostly known as large intelligent surfaces with a large number of elements on it \cite{Ertrul6G,wcnc}. 
Therefore, if the number of elements is large enough, profiting from central limit theorem, 
the PMF of $\Gamma$ can be approximated by a Gaussian distribution function with mean $\mathcal{M}$ and variance of $\sigma_z^2$ as follows.
    \begin{equation}
        f_\Gamma(T)=\frac{1}{\sqrt{2\pi\sigma_z^2}}\exp\left[-\frac{(T-\mathcal{M})^2}{2\sigma_z^2}\right].
    \end{equation}
    Eventually, the coverage probability becomes 
    \begin{equation}
        \Pr\left[\Gamma>T\right]=\mathcal{Q}\left(\frac{T-\mathcal{M}}{\sigma_z}\right).\nonumber
    \end{equation}
\end{proof}

The accuracy of this approximation is quiet satisfactory for large RISs since the weights of $\mathcal{A}_n$ hold specific values with small variance 
\cite{sumind}. Fig. \ref{ProdDis1} shows that this PMF approximation coincides the simulation results for large RISs.

 However, the accuracy of \textbf{Approximation-\rom{1}} becomes poor for a small $N$. In such a rare case where $N$ is small, e.g., $N<20$, we are able to find out the outage probability using numerical techniques to take into account all combinations. 
    \subsection{\textbf{Approximation-\rom{2}}}
    In this approximation, we restrict $\mathcal{A}_n$ in \eqref{rt22} to take an average of SNR coverage probability, specially when $N$ is small and \textbf{Approximation-\rom{1}} might not be applicable. Therefore, we state the following proposition.
    
    \begin{proposition}
    The average O2I  SNR  coverage  probability  when  $\mathcal{A}_n$ in \eqref{rt22} is restricted to be fixed for all $N$ paths, i.e., $\forall n \quad \rightarrow  \mathcal{A}_n\approx\mathcal{A}=\mathbb{E}\{\mathcal{A}_n\}$ is approximated by
     \begin{align}
         &\Pr\left[\Gamma_\emph{{\rom{3}}}>T|T=k\mathcal{GA}\right]>\nonumber\\
         & 1- \sum_{q=0}^k \frac{1}{N+1}\sum_{\ell=0}^N \mathcal{C}^{-k\ell} \prod_{n=1}^N\left[ 1+(\mathcal{C}^\ell-1)(1-\mathfrak{p}_n)\right],
    \end{align}
     where
     \begin{align}
     \Gamma_\emph{\text{\rom{3}}} =\mathcal{GA}\sum_{n=0}^{N-1}Z_n
         \hspace{0.5cm}\text{and}\hspace{.5cm}
         \mathcal{C}=\exp \left[\frac{2j\pi}{N+1}\right]
          .
        \label{rt23}
     \end{align}
     
    \end{proposition}
    \begin{proof}
    The proof of this proposition is given in Appendix.
    \end{proof}

    \begin{figure}[t]
            \centering
                \includegraphics[width=1\linewidth-1cm, height=5cm]{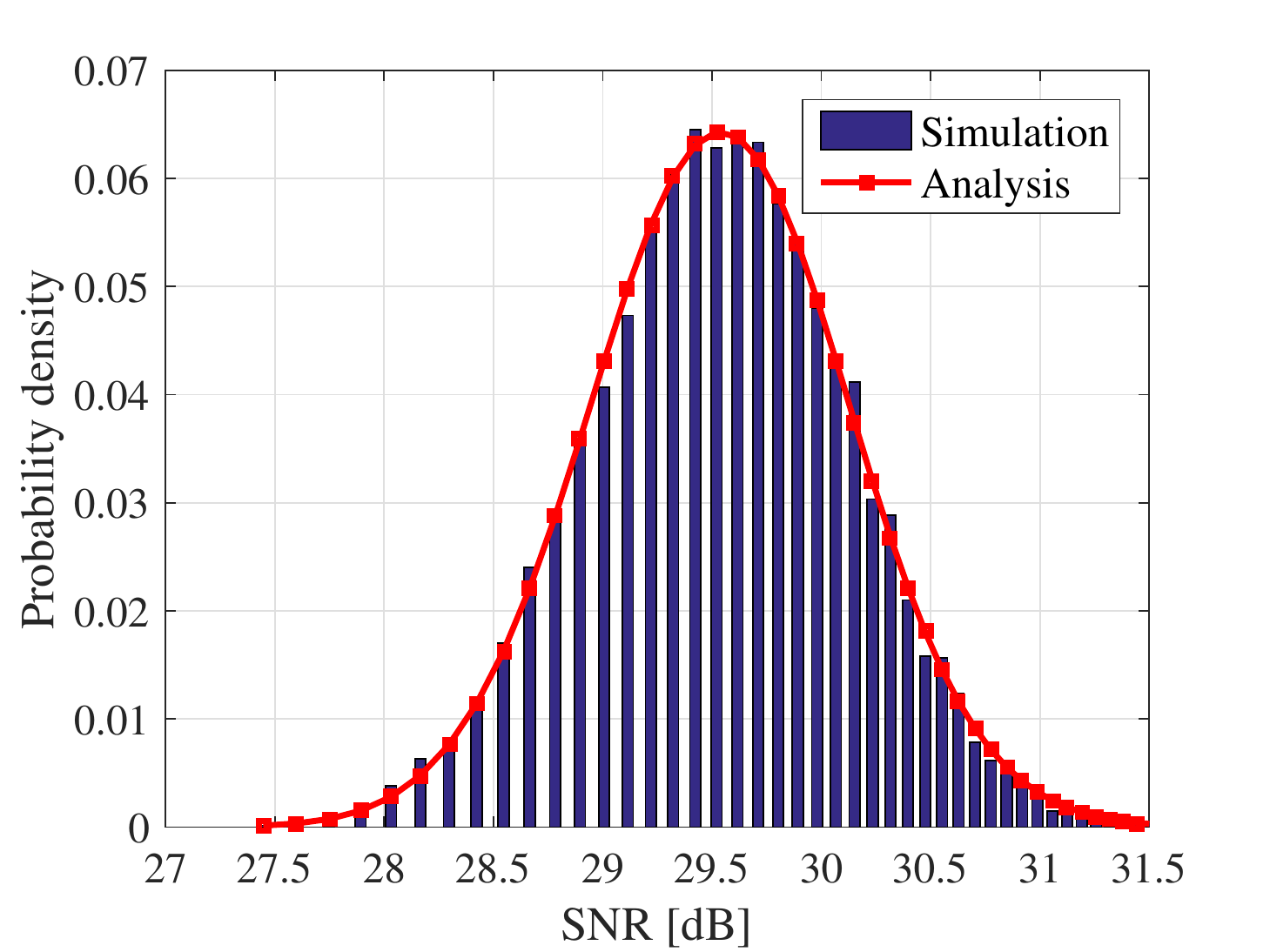}
                \caption{\textbf{Approximation \rom{1}}, PMF of $\Gamma$ when $N=64$.}
            \label{ProdDis1}
    \end{figure}

 Apart from above approximations, using Chernoff bound and profiting from theorems 1 and 2 in \cite{WeiBook}, we respectively have upper Chernoff bound on the SNR coverage probability as follows.
    \begin{align}
    \left\{\begin{array}{l}
    \Pr\left[\Gamma>T|T<\mathcal{M}\right]<1-\left(2-T/\mathcal{M}\right)^{T-2\mathcal{M}}e^{\left(\mathcal{M}-T\right)}\vspace{2mm}\\
        \Pr\left[\Gamma>T|T>\mathcal{M}\right]<\left(T/\mathcal{M}\right)^{-T}e^{\left(T-\mathcal{M}\right)} 
    \end{array}
        \right.
    \end{align}

%% file: Sections/3_results.tex
In this section, the performance of the proposed method under different scenarios is evaluated. The evaluations include A) SNR coverage comparison between RIS-assisted model and its counterparts, B) Impact of blockage density on SNR coverage performance, and C) Impact of RIS size and blockage probability on SNR.

We use MATLAB and consider a UPA RIS-empowered wall with area of $20\times 20$ $\text{m}^2$; where the rest of the parameters used in the simulations are given in TABLE \ref{tab1} unless otherwise specified.

\begin{table}[h]
\begin{center}
\centering
\captionsetup{width=1\linewidth}
\caption{System numerical parameters.\label{tab1}}
\label{tab:par}
\begin{tabular}{lc} 
\hline
System parameters & Corresponding value  \\ 
 \hline
 \hline
 Signal power (downlink)  &  30 dBm\\
 \hline
  Pathloss exponent, $\alpha$  &  $4$\\
 \hline
BS antenna array size, $M$ & 64 \\
 \hline
Noise floor power & -110 dBm\\
\hline
Small-scale fading& $\sim$\textit{Nakagami}$(3,1)$\\
\hline
Sensor attenuation factor, $B_n$ & 0.9\\
%
\hline

\end{tabular}
\end{center}
\end{table}

\subsection{SNR coverage comparison between RIS-assisted model and its counterparts}

Fig. \ref{fig:33} compares the RIS-assisted model and its two counterparts in terms of the SNR coverage performance. For this set of simulations, the location of the outdoor BS is fixed in $60$ m distance from the intent building wall empowered by RIS-sensors and the indoor UE is in $10$ m distance of the RIS wall. The height of the BS and the building level where the UE is located are set as $200$ and $100$ meters, respectively. Among different building materials, glass windows are considered critical materials
for designing and optimizing O2I coverage since their penetration loss is less than that of other building materials which depends on
the composition, thickness, and layers of windows and varies
greatly over frequency \cite{OUT3,RAP}. We consider a clear glass with a penetration loss of $3.6$ dB in 28 GHz according the model in \cite{RAP}.  
Moreover, for the relay-assisted model we use the existing model in \cite{OUT1}. For outdoor blockage: $\lambda_\text{st}^\text{out}=25\frac{\text{bl}}{\text{km}^2}$, 
    $\mathbb{E}\{\mathcal{W}\}=\mathbb{E}\{\mathcal{L}\}=10$m, $\eta_1=0.5$, and for indoor blockage:  $\lambda_\text{st}^\text{in}=0.1\frac{\text{bl}}{\text{m}^2}$,
  $\lambda_\text{dy}^\text{in}=0.1\frac{\text{bl}}{\text{m}^2}$,
  $H^\text{in}_\text{bl}=2$m, $\mathbb{E}\{\mathcal{W}\}=\mathbb{E}\{\mathcal{L}\}=0.5$m, 
  $\eta_2=0.25$,
   $\mu_\text{in}=1$s, $V_\text{in}=0.5\frac{\text{m}}{\text{s}}$.

As shown in Fig. \ref{fig:33}, the RIS-assisted model outperforms both of its counterparts. In other words, the RIS-assisted O2I communication model overcomes the penetration loss due to the building materials. Furthermore, the proposed RIS-assisted model provides a diversity gain due to the wide surface of the RIS. As a result, the blockage probability becomes lower than that of the similar relay-aided counterpart in \cite{OUT1}.
Additionally, as the number of RIS-elements increases, the SNR coverage performance increases. We can see that the \textbf{Approximation-\rom{1}} agrees with the simulation results especially when $N$ is relatively large. Besides, \textbf{Approximation-\rom{2}} gives a wide range of the SNR coverage performance when $N$ is relatively small, e.g., $N=9$. It is evident that in such a small $N$ case, more precise approximation of the outage probability using numerical techniques by taking into account all combinations can be obtained.

\begin{figure}[t]
\centering
\captionsetup{width=1\linewidth}
  \includegraphics[width=8.5cm, height=5.5cm]{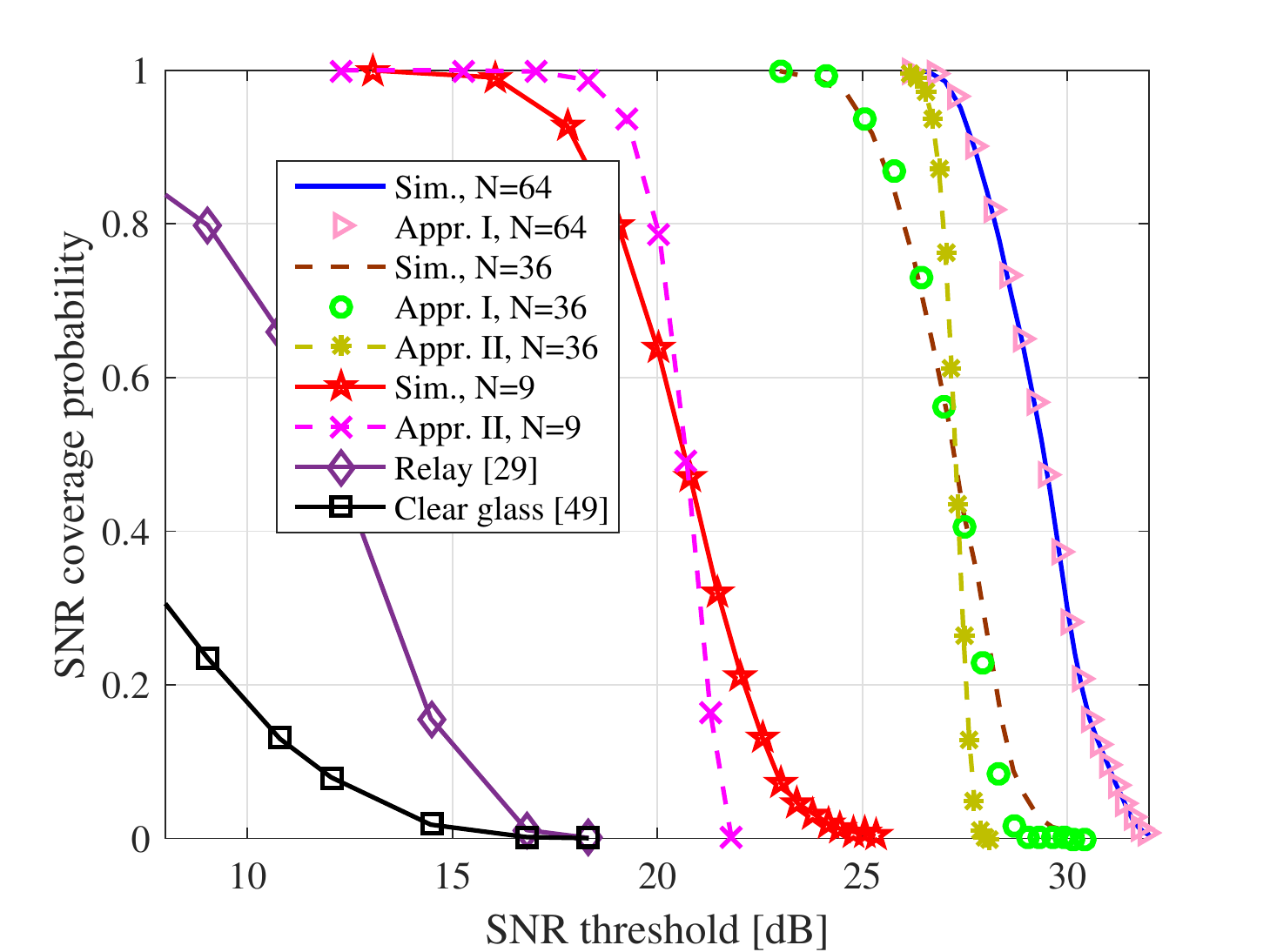}%
  \caption{SNR coverage performance w.r.t. the SNR threshold.}
  %
  \label{fig:33}
\end{figure}
\subsection{Impact of blockage density on SNR coverage performance}
This subsection evaluates the SNR coverage performance of the proposed RIS-assisted O2I communication model when the indoor static blockage density increases. For this set of simulations, the blockage parameters are the same as those in the previous subsection. As shown in Fig. \ref{fig:34}, when $\lambda_\text{st}^{\text{in}}$ increases, the SNR coverage probability decreases. This happens because the increase of $\lambda_\text{st}^\text{in}$ results in the increase of blockage probability as shown in Fig. \ref{ProdDis0}. Likewise, an increase of either $\lambda_\text{st}^\text{out}$ or $\lambda_\text{dy}^\text{in}$ has the same impact on the SNR coverage performance of the proposed model. 


\subsection{Impact of RIS size and blockage probability on SNR}
As the third evaluation, Fig. \ref{fig:35} shows the impact of the number of RIS-sensors (i.e., RIS-size) and the blockage probability on the SNR coverage performance of the RIS-assisted O2I communication. As it is shown, $N$ has a significant impact on the SNR performance and can overcome the blockage when RIS-size increases. In this simulation, the area of RIS wall is set to $20\times20$ $\text{m}^2$. However, in worse conditions that obstacles size is big and the blockage correlation between the paths is high, further increase of RIS-sensors may require to accompany with a wider RIS-wall to overcome blockage and also satisfies the assumption in \textit{Remark 1}. 
Nevertheless, there might be practical constraints that restrict the further increase of RIS size which is a bottleneck of the proposed idea but it is still a better solution compared to its counterparts in the similar blockage situations.
\begin{figure}[t]
\centering
\captionsetup{width=1\linewidth}
  \includegraphics[width=8.5cm, height=5.5cm]{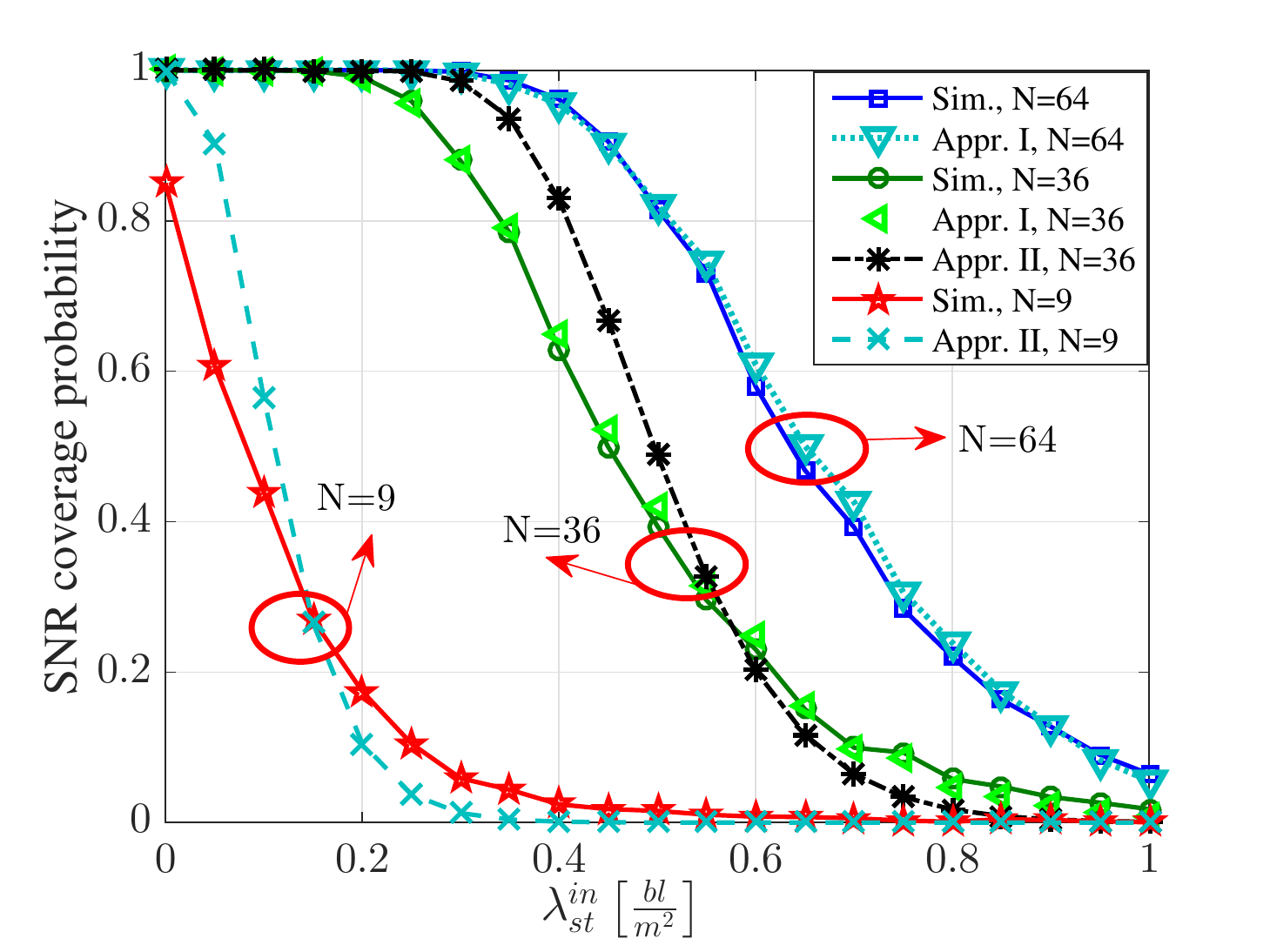}%
  \caption{SNR coverage performance w.r.t. the $\lambda_{\text{st}}^{\text{in}}$ when $T=20$ dB.}
  \label{fig:34}
\end{figure}

%% file: Sections/Conc.tex

We developed  RIS empowered wall  for  intelligent  transition  of  the signal from outdoor to indoor space to ameliorate the  bottleneck  of existing  O2I  mmWave communication. We  exploited the metasurface design for O2I refractor based on our new idea of developing wall-based fabrication consisting  of  several  chipless RFID sensors into the building wall using outdoor and indoor antennas, which can potentially  form  a  system  that  will  transport  the signal  from  the O2I  space. The chipless RFID sensor used in the design includes a bank of delay lines that can be controlled by the RIS phase-shift controller to adjust the phase-shift at the sensor and execute passive beamforming at the RIS wall towards the indoor UE. We used  mathematics  as  a  tool for  reasoning  the  feasibility  of  the  proposed  system  with approximations for the signal strength and O2I communication and evaluating the blockage effect of the developed system.

\begin{figure}[t]
\centering
\captionsetup{width=1\linewidth}
  \includegraphics[width=8.5cm, height=5.5cm]{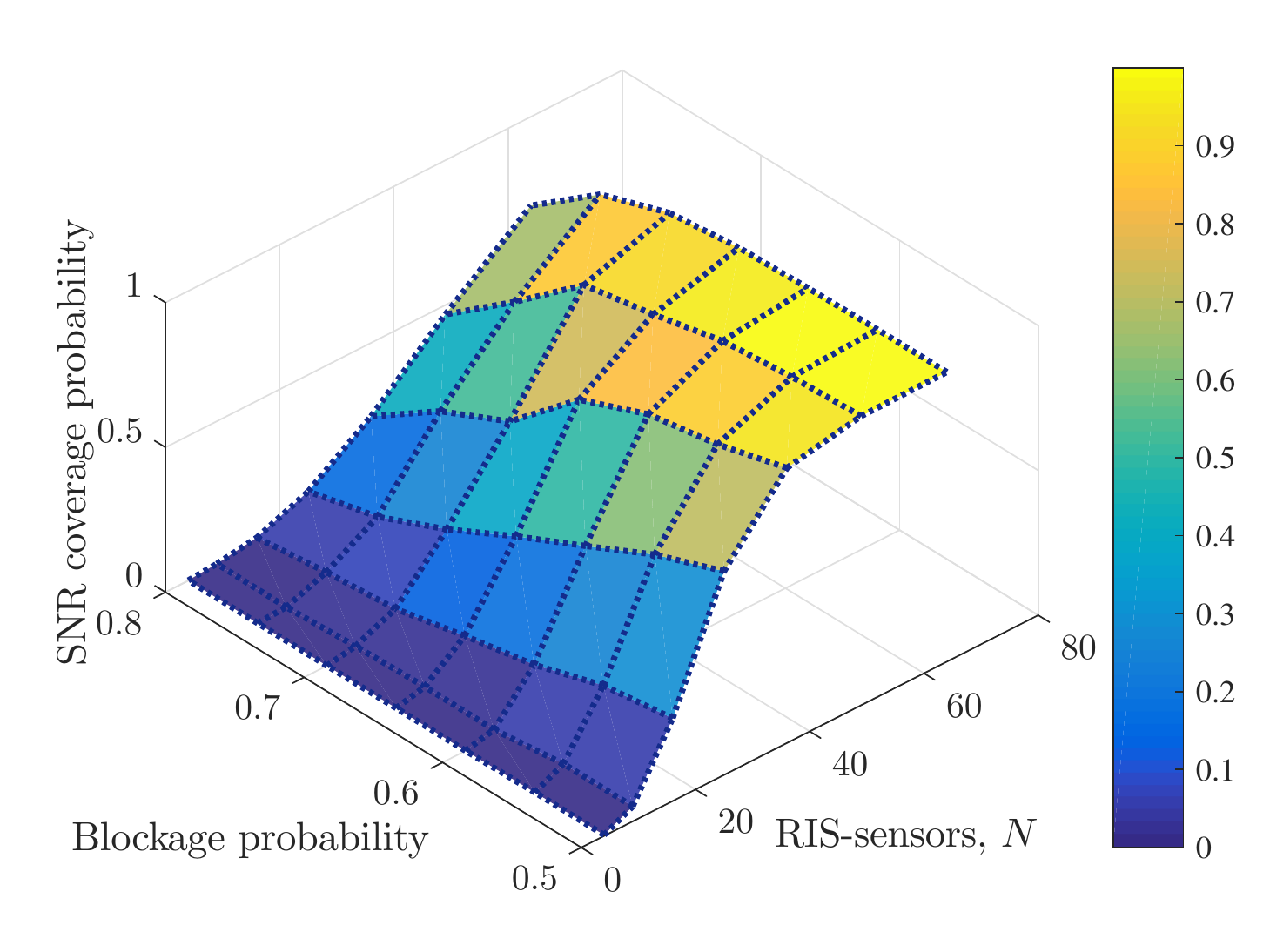}%
  \caption{SNR coverage performance w.r.t. the $p$ and $N$ when $T=22$ dB.}%
  \label{fig:35}
\end{figure}

%% file: Sections/5_app.tex
 \label{app2}
 PMF analysis of $\Gamma$ in \eqref{rt22} for small $N$ is quite complicated. Here we aim at restricting $\mathcal{A}_n$ to simplify the analysis when $N$ is small. Here, it is a well-known fact that the small-scale fading at mmWave bands is much less severe than that in low/medium frequency bands due to narrow beamforming \cite{Cov&rate}. Thus, a large \textit{Nakagami-m} parameter can be used to approximate the small-variance LoS fading.
     Moreover, thanks to the $B_n$ factor and having close and limited number of links, they result in stationary and relatively equivalent amplitudes of $\mathcal{A}_n$. 
   Thus, under these special circumstances, we can assume that the difference between the amplitudes of $\mathcal{A}_n$ coefficients is negligible (i.e., the path gains vary slowly) and consider them as one recommended average value as follows.
    \begin{equation}
        \forall n \quad \rightarrow \mathcal{A}_n\approx \mathbb{E}\{\mathcal{A}_n\}= \mathcal{A}.
    \end{equation}
        Consequently, the expression in \eqref{rt22} can be rewritten as
    \begin{equation}
        \Gamma_\text{\rom{3}} =\mathcal{GA}\sum_{n=0}^{N-1}Z_n,
        \nonumber
    \end{equation}
    with a Poisson binomial distribution. Thus, its closed-form PMF \cite{wang1993number} becomes
    \begin{equation}
        \Pr\{\Gamma_{\text{\rom{3}}}=k \mathcal{GA}\}=\sum_{\mathbf{G}\in \mathbf{I}_k} \, \, \prod_{\imath\in \mathbf{G}} (1-\mathfrak{p}_\imath) \prod_{\jmath\in \mathbf{G}^c} \mathfrak{p}_\jmath,
        \label{gam}
    \end{equation}
    where $\mathbf{I}_k$ is the set of all subsets of $k$ integers that can be selected from $\{1,\cdots,N\}$ and $\mathbf{G}^c$ is the complement of $\mathbf{G}$. Therefore, the coverage probability when RIS contains low number of elements becomes
     \begin{equation}
         \Pr\left[\Gamma_\text{\rom{3}}>T|T=k\mathcal{GA}\right]=1-\sum_{q=1}^k \sum_{\mathbf{G}\in \mathbf{I}_k} \, \, \prod_{\imath\in \mathbf{G}} (1-\mathfrak{p}_\imath) \prod_{\jmath\in \mathbf{G}^c} \mathfrak{p}_\jmath.
        \label{gam1}
     \end{equation}
      $\mathbf{I}_k$ contains $\frac{N!}{(N-k)!k!}$ elements which results in computation complexity in practice, but still suitable when $N$ is small. 
     Besides, the PMF in \eqref{gam} can be simplified utilizing discrete Fourier transform (DFT) \cite{DFT} as follows
    \begin{align}
         &\Pr\{\Gamma_\text{\rom{3}}=k \mathcal{GA}\}= \nonumber\\
         &\frac{1}{N+1}\sum_{\ell=0}^N \mathcal{C}^{-k\ell} \prod_{n=1}^N\left[ 1+(\mathcal{C}^\ell-1)(1-\mathfrak{p}_n)\right],
    \end{align}
    where $\mathcal{C}=\exp \left[\frac{2j\pi}{N+1}\right]$. Therefore, the SNR coverage probability can be re-expressed as
    \begin{align}
         &\Pr\left[\Gamma_\text{\rom{3}}>T|T=k\mathcal{GA}\right]=\nonumber\\
         & 1- \sum_{q=0}^k \frac{1}{N+1}\sum_{\ell=0}^N \mathcal{C}^{-k\ell} \prod_{n=1}^N\left[ 1+(\mathcal{C}^\ell-1)(1-\mathfrak{p}_n)\right]\nonumber.
    \end{align}
    The analysis is complete.